\documentclass[pra,twocolumn,floatfix,superscriptaddress,longbibliography,notitlepage]{revtex4-1}

\usepackage{graphicx, color, graphpap}

\usepackage{comment}
\usepackage{enumitem}
\usepackage{amssymb}
\usepackage{amsthm}
\usepackage{multirow}
\usepackage[colorlinks=true,citecolor=blue,linkcolor=magenta]{hyperref}
\usepackage[T1]{fontenc}
\usepackage{bbm}
\usepackage{thmtools,thm-restate}
\usepackage{verbatim}
\usepackage{mathtools}
\usepackage{titlesec}
\usepackage{amsmath}
\usepackage[title]{appendix}
\usepackage[linesnumbered,ruled,vlined]{algorithm2e}
\SetKwInput{kwInit}{Init}
\usepackage{placeins}

\long\def\ca#1\cb{} 

\newcommand{\norm}[2][]{#1| \! #1| #2 #1| \! #1|}

\newcommand{\poly}{\operatorname{poly}}

\newcommand{\DC}{\mathcal{D}}
\newcommand{\EC}{\mathcal{E}}

\newcommand{\LC}{\mathcal{L}}

\newcommand{\OC}{\mathcal{O}}

\newcommand{\RC}{\mathcal{R}}
\newcommand{\SC}{\mathcal{S}}

\newcommand{\Tr}{{\rm Tr}}

\renewcommand{\geq}{\geqslant}
\renewcommand{\leq}{\leqslant}

\DeclareMathOperator*{\argmin}{arg\,min}
\renewcommand{\vec}[1]{\boldsymbol{#1}}  

\newcommand*{\id}{\openone}

\newcommand{\tin}{{\text{in}}}

{}
{}
\newtheorem{theorem}{Theorem}
\newtheorem{lemma}{Lemma}
\newtheorem{corollary}{Corollary}

\usepackage{graphicx}
\usepackage{dcolumn}
\usepackage{bm}
\usepackage{hyperref}

\usepackage{xcolor}
\hypersetup{colorlinks=true}

\begin{document}
\title{Inference-Based Quantum Sensing}

\author{C. Huerta Alderete}
\thanks{The two first authors contributed equally.}
\affiliation{Information Sciences, Los Alamos National Laboratory, Los Alamos, NM 87545, USA}
\affiliation{Materials Physics and Applications Division, Los Alamos National Laboratory, Los Alamos, NM 87545, USA.
}
\affiliation{Quantum Science Center, Oak Ridge, TN 37931, USA}

\author{Max Hunter Gordon}
\thanks{The two first authors contributed equally.}
\affiliation{Theoretical Division, Los Alamos National Laboratory, Los Alamos, NM 87545, USA}
\affiliation{Instituto de Física Teórica, UAM/CSIC, Universidad Autónoma de Madrid, Madrid 28049, Spain}

\author{Fr\'{e}d\'{e}ric Sauvage}
\affiliation{Theoretical Division, Los Alamos National Laboratory, Los Alamos, NM 87545, USA}

\author{Akira Sone}
\affiliation{Aliro Technologies, Inc, Boston, MA 02135, USA}

\author{Andrew T. Sornborger}
\affiliation{Information Sciences, Los Alamos National Laboratory, Los Alamos, NM 87545, USA}
\affiliation{Quantum Science Center, Oak Ridge, TN 37931, USA}

\author{Patrick J. Coles}
\affiliation{Theoretical Division, Los Alamos National Laboratory, Los Alamos, NM 87545, USA}
\affiliation{Quantum Science Center, Oak Ridge, TN 37931, USA}

\author{M. Cerezo}
\email{cerezo@lanl.gov} 
\affiliation{Information Sciences, Los Alamos National Laboratory, Los Alamos, NM 87545, USA}
\affiliation{Quantum Science Center, Oak Ridge, TN 37931, USA}

\begin{abstract}
In a standard Quantum Sensing (QS) task one aims at estimating an unknown parameter $\theta$, encoded into an $n$-qubit probe state, via measurements of the system. The success of this task hinges on the ability to correlate changes in the parameter to changes in the system response $\RC(\theta)$ (i.e., changes in the measurement outcomes). For simple cases the  form of $\RC(\theta)$ is known, but the same cannot be said for realistic scenarios, as no general closed-form expression  exists. In this work we present an inference-based scheme for QS. We show that, for a general class of unitary families of encoding, $\RC(\theta)$  can be fully characterized by only measuring the system response at $2n+1$ parameters. This allows us to infer the value of an unknown parameter given the measured response, as well as to determine the sensitivity of the scheme, which characterizes its overall performance. We show that  inference error is, with high probability, smaller than $\delta$, if one measures the system response with a number of shots that scales only as $\Omega(\log^3(n)/\delta^2)$.  Furthermore, the framework presented can be broadly applied as it remains valid for arbitrary probe states and measurement schemes, and, even holds in the presence of quantum noise. We also discuss how to extend our results beyond unitary families. Finally, to showcase our method we implement it for a QS task on real quantum hardware, and in numerical simulations.  
\end{abstract}

\maketitle

\textit{Introduction.} Quantum Sensing (QS) is one of the most promising applications for  quantum technologies~\cite{degen2017quantum}.  In QS experiments one uses a quantum system as a probe to interact with an environment. Then, by measuring the system, one aims at learning some relevant property of the environment (usually some characteristic parameter) with a precision and sensitivity that are higher than those achievable by any classical system~\cite{giovannetti2006quantum}. QS has applications in a wide range of fields such as  quantum magnetometry~\cite{taylor2008high, bhattacharjee2020quantum, barry2016optical,casola2018probing}, thermometry~\cite{correa2015individual,de2016local,sone2018quantifying,sone2019nonclassical},  dark matter detection~\cite{rajendran2017method}, and gravitational wave detection~\cite{mcculler2020frequency,tse2019quantum}.

In a QS experiment one first prepares an $n$-qubit probe state $\rho$ that is as sensitive as possible to an external parameter $\theta$ of interest.  This ensures that upon encoding two distinct parameters $\theta$ and $\theta'$ on the system, the respective measurements associated to $\rho_{\theta}$ and $\rho_{\theta'}$  will be sufficiently distinguishable, a prerequisite to any task of sensing. 
Second, one  obtains the system response $\RC(\theta)$ to the external interaction by  measuring some observable over $\rho_\theta$. Third, if the functional form of $\RC(\theta)$ is known and invertible,  one can infer the value of $\theta$ from measurement outcomes, as well as estimate the sensitivity of the  QS scheme.

In simple cases all the  previous steps are well characterized. For instance, in an idealized magnetometry experiment it is known that the optimal probe state is the $n$-qubit Greenberg-Horne-Zeilinger (GHZ) state, while the optimal measurement is a parity measurement~\cite{greenberger1990bells,leibfried2004toward}. In this case, $\RC(\theta)=\cos(n \theta)$ which allows one to obtain the magnetic field as $\theta=\cos^{-1}(\RC(\theta))/n$ (assuming $\theta\in(-\pi/n,\pi/n)$), and the state's sensitivity as $(\Delta\theta)^2=1/n^2$, which corresponds to  the Heisenberg limit~\cite{giovannetti2006quantum}. However, the situation becomes more involved in realistic scenarios where the system dynamics are not known, and hence where the explicit functional form of $\RC(\theta)$ may not be accessible. For instance, when noise in the magnetometry setting is taken into account, the GHZ state is no longer optimal~\cite{huelga1997improvement,koczor2020variational,fiderer2019maximal}. In this case the true response $\RC(\theta)$ will inevitably deviate from the idealized cosine formula, limiting the extent to which $\theta$ can be accurately estimated. While recent works have focused on maximizing the sensitivity of QS protocols in noisy situations, by means of variational approaches~\cite{cerezo2020variationalreview,koczor2020variational,beckey2020variational,sone2020generalized,cerezo2021sub,kaubruegger2021quantum,meyer2020variational}, methods to recover the true  $\RC(\theta)$ in-situ are still lacking.

Here we introduce a  data-driven inference method which allows us to efficiently characterize the exact functional form of $\RC(\theta)$ for a general class of unitary families. We show that $\RC(\theta)$ can be expressed as a trigonometric polynomial of degree $n$, such that it can be fully determined by only measuring the system response at a set of $2n+1$ known parameters. We then discuss how the inferred function can be used to estimate the value of \textit{any} unknown parameter, as well as the sensitivity of the  scheme. Moreover, we rigorously analyze the inference error. Finally, we show that our method can be extended to cases where the system response is no longer exactly a trigonometric polynomial, but can still be approximated by one.  The  applications of the inference scheme are demonstrated in both numerical simulations as well as real implementations on a quantum computer.

\textit{Results.} Here  we consider a single-parameter QS setting employing an $n$-qubit probe state $\rho$ to estimate a parameter $\theta$. As shown in Fig.~\ref{fig:1}, $\rho$ is prepared by sending a fiduciary state $\rho_{\tin}$ through a state preparation channel $\EC$ such  that $\EC(\rho_\tin)=\rho$. We focus on the case of unitary families where the parameter encoding mechanism is of the form
\begin{equation}\label{eq:unitary_family}
   \mathcal{S}_\theta(\rho) = e^{-i \theta H/2}\rho ~e^{i \theta H/2}= \rho_{\theta}\,.
\end{equation}
Here, $H$ is a Hermitian operator such that $H=\sum_j h_j$ with $h_j^2=\id$, and $[h_j,h_{j'}]=0$, $\forall j, j'$. As shown  below, the Hamiltonian in a magnetometry task is precisely of this form.  We allow for the possibility of sending $\rho_\theta$ through a second pre-measurement channel $\DC$ which outputs an $m$-qubit state $\DC(\rho_\theta)$ (with $m\leq n$), over which we measure the expectation value of an observable $O$, with $\norm{O}_\infty\leq 1$. The system response is thus defined as 
\begin{align}\label{eq:response}
    \RC(\theta)&=\Tr[\DC\circ\SC_{\theta}\circ\EC(\rho_\tin)O]\,.
\end{align}
This setting encompasses cases where $\EC$ or $\DC$ are noisy channels, as well as cases of imperfect parameter encoding where a $\theta$-independent noise channel acts after $\mathcal{S}_\theta$, as is further discussed  in the Supplemental Material (SM) \footnote{See Supplemental Material which contains additional details and proofs as well as Refs. \cite{hong1992lower,horn1991topics, jackson1913accuracy,petras1995error} }.

Leveraging tools from the quantum machine learning literature~\cite{nakanishi2020sequential} we prove the following theorem.
\begin{theorem}\label{theo:trig-pol}
Let $\RC(\theta) $ be the response function in Eq.~\eqref{eq:response} for a unitary family as in Eq.~\eqref{eq:unitary_family}. Then, for any $\EC$, $\DC$ and measurement operator $O$, $\RC(\theta) $ can be exactly expressed as a trigonometric polynomial of degree $n$. That is, 
\begin{align}\label{eq:trig_response}
    \RC(\theta) &= \sum_{s=1}^n \left[a_{s} \cos(s\theta) + b_{s} \sin(s\theta)\right] + c\,,
\end{align} 
with $\{a_s,b_s\}_{s=1}^n$ and $c$ being real valued coefficients.
\end{theorem}
Notably, Theorem~\ref{theo:trig-pol} determines the \textit{exact} functional relation between the encoded parameter $\theta$ and the system response. Furthermore, the $2n+1$ coefficients  $\{a_s,b_s\}_{s=1}^n$ and $c$, that are not known a priori, can be efficiently estimated by means of a trigonometric interpolation~\cite{zygmund2002trigonometric}. This is readily achieved by measuring the system responses at a set of predefined parameters $P=\{\theta_k\}_{k=1}^{2n+1}$ (see Fig.~\ref{fig:1}), as this leads to a system of $2n+1$ equations with $2n+1$ unknown variables. Hence, one needs to solve a linear system problem of the form $A\cdot\boldsymbol{x}=\boldsymbol{d}$. Here,  $\vec{x} =(a_1,\ldots a_{n}, b_1,\ldots b_{n}, c)$ is the vector of unknown coefficients, $\vec{d} = (\RC(\theta_1),  \ldots, \RC(\theta_{2n+1}))$ is a vector of measured system responses across $P$ and $A$ is a $(2n+1)\times(2n+1)$ matrix with elements $A_{kj}= \cos{(j \theta_{k})}$ for $j=1,\ldots n$, $A_{kj}= \sin{(j \theta_{k})}$ for $j=n+1,\ldots2n$ and $A_{k(2n+1)}=1$. Thus, solving $\vec{x}=A^{-1}\cdot\vec{d}$ allows us to fully characterize $\RC(\theta)$. In the SM we provide additional details on this linear system problem.

\begin{figure}[t!]
	\includegraphics[width= 1 \columnwidth]{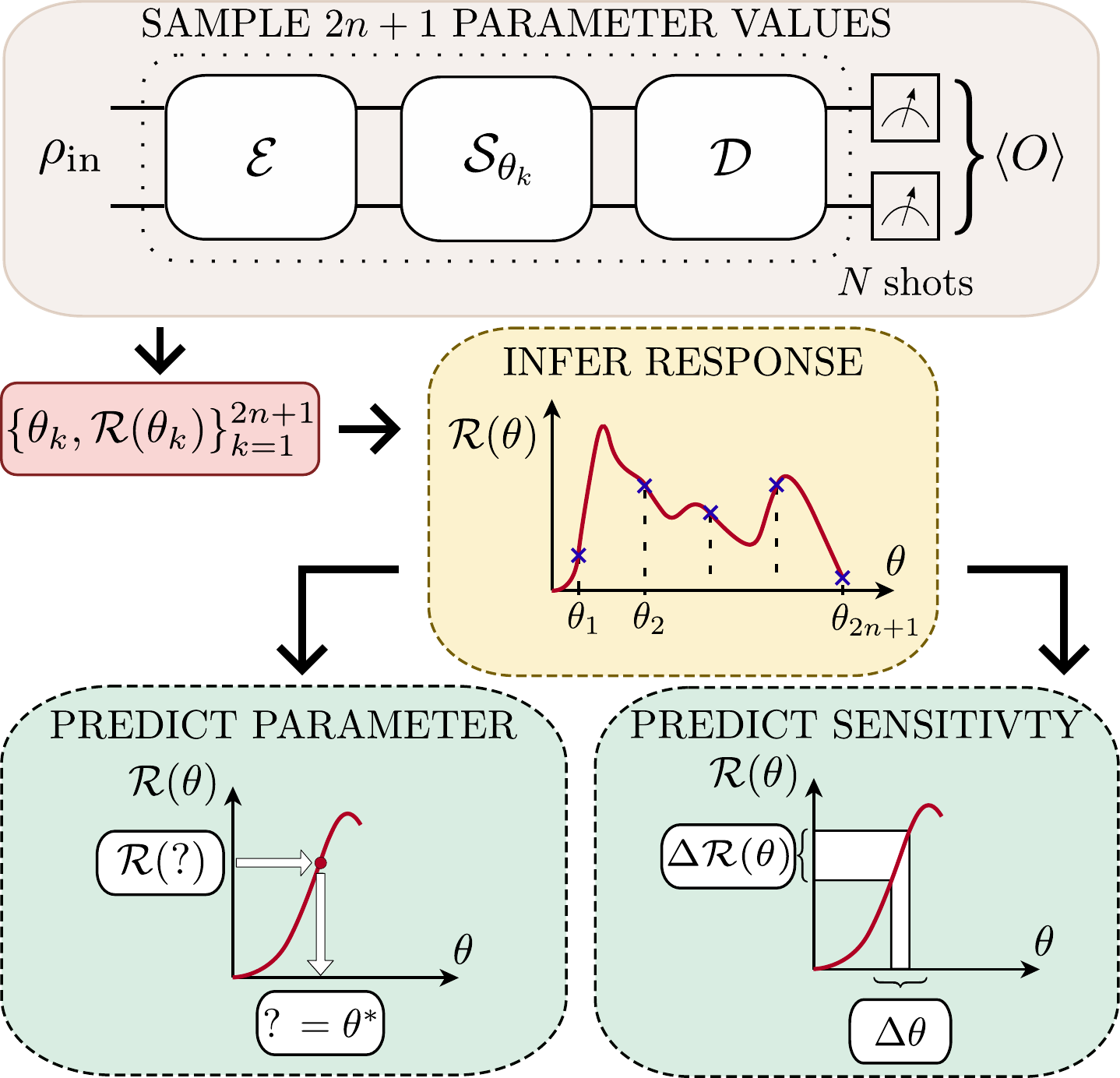}
	\caption{\textbf{Inference-based QS scheme.} An input state $\rho_{\tin}$ is sent through the following channels: state preparation $\EC$, parameter encoding $\SC_\theta$, pre-measurement $\DC$. We then measure the expectation value of  $O$. By measuring the system response at $2n+1$ parameters, we can recover the exact form of $\RC(\theta)$ in Eq.~\eqref{eq:trig_response}. From  $\RC(\theta)$ we can  infer the value of a parameter given the system response, and   compute the sensitivity of the sensing scheme.  }
	\label{fig:1}
\end{figure}

\begin{figure*}[t!]
    \centering
    \includegraphics[width = 2\columnwidth]{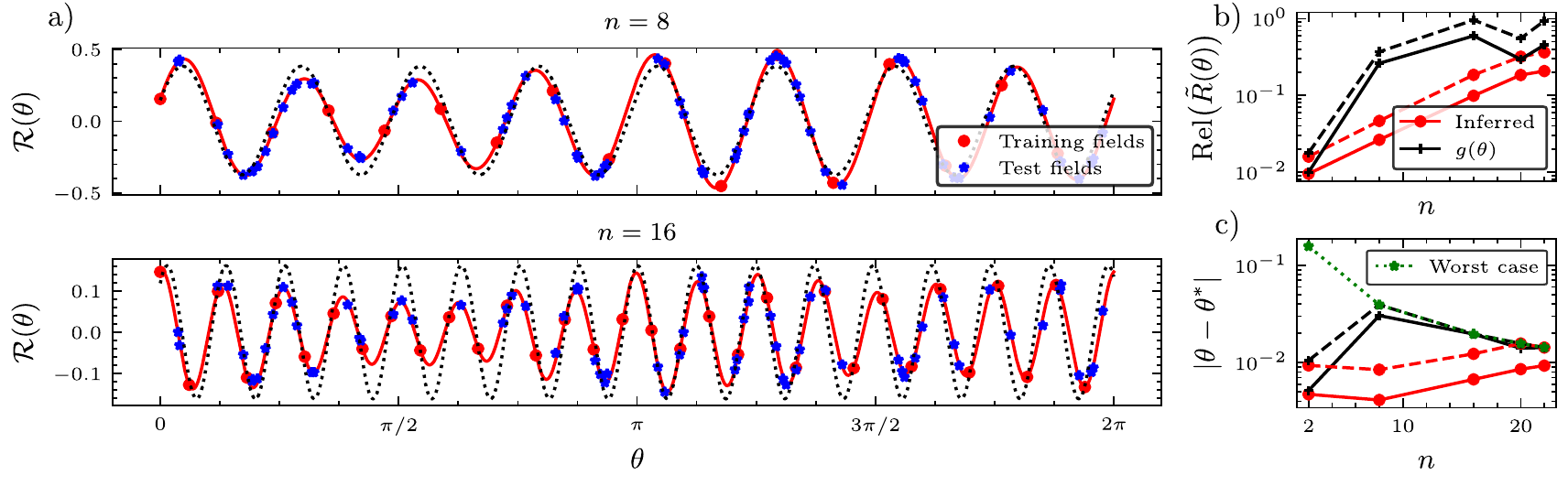}
    \caption{\textbf{Magnetometry task on IBM hardware}. a) Inferred response $\widetilde{\RC}(\theta)$ for system sizes of $n=8$ and $16$ qubits. The fields used to train (red point) and test (blue star) the inference scheme, were estimated on the \textit{IBM\_Montreal} quantum computer using $3.2 \times 10^{4}$ shots per expectation value. We depict the inferred response $\widetilde{\RC}(\theta)$ (red solid curve) as well as the  fit $g(\theta)= \alpha\cos(\beta \theta+ \gamma)+\zeta$ (black dotted curve).
    b) Relative response error versus $n$. Statistics were obtained over $74$ test fields and $7$ experiment repetitions. The relative error is defined as the difference between the fit or inferred value and the measurement response, normalized by the average test expectation value. The  red (black) points  correspond to  $\widetilde{\RC}(\theta)$ ($g(\theta$)), while solid (dashed) lines represent the median (upper quartile) error. c) Parameter prediction error  versus $n$, with green dots denoting the worst possible prediction (see SM).  }
    \label{fig:real_device_example_error_scaling}
\end{figure*}

Here we note that $A$ can be singular (for instance if $\theta_k=\theta_{k'}+2\pi$ for any two $\theta_k,\theta_{k'}\in P$), and hence care must be taken when determining the $2n+1$ parameters. As shown in the SM, the optimal strategy is to uniformly sample the parameters  as
\begin{equation}\label{eq:sampling}
   \theta_k=\frac{2 \pi (k-1)}{2n+1}\,,\quad \text{with $k=1,\ldots,2n+1$}\,,
\end{equation}
since this choice reduces the matrix inversion error. 

In practice one cannot exactly evaluate the responses $\RC(\theta_k)$, but rather can only estimate them up to some statistical uncertainty resulting from finite sampling. We define  $\overline{\RC}(\theta_k)$ as the $N$-shot estimate of $\RC(\theta_k)$, and   $\widetilde{\RC}(\theta)$ as the inferred response, a trigonometric polynomial of the form in Eq.~\eqref{eq:trig_response}, obtained by solving the linear system of equations with the estimates $\overline{\RC}(\theta_k)$. 
The effect of the estimation errors on the accuracy of the inferred response can be rigorously quantified as follows.
\begin{theorem} \label{theo:bound1}
Let $\RC(\theta)$ be the exact response function, and let $\widetilde{\RC}(\theta)$ be its approximation obtained from the $N$-shot estimates $\overline{\RC}(\theta_k)$ with $\theta_k$ given by Eq.~\eqref{eq:sampling}.  Defining the maximum estimation error $\varepsilon=\max_{\theta_k\in P}|\RC(\theta_k)-\overline{\RC}(\theta_k)|$, then we have that for all $\theta$
\begin{equation}
    \vert \RC(\theta)-\widetilde{\mathcal{R}}(\theta) \vert \in\OC\left(\varepsilon\log(n)\right)\,.
\end{equation}
\end{theorem}

Since the maximum estimation error $\varepsilon$ is fundamentally related to the number of shots $N$, we can derive the following corollary.

\begin{corollary}\label{cor:error}
The number of shots $N$, necessary to ensure that with a (constant) high probability, and for all $\theta$, the error $\vert \RC(\theta)-\widetilde{\mathcal{R}}(\theta) \vert\leq \delta$, for an inference  error $\delta$, is in $\Omega\left(\log^3(n)/\delta^2\right)$.
\end{corollary}

It follows from Corollary~\ref{cor:error} that, for fixed $\delta$, a poly-logarithmic number of shots $N\in\Omega(\log^3(n))$ suffices to guarantee that $\widetilde{\mathcal{R}}(\theta)$ will be a good approximation for the true response $\RC(\theta)$. Once the inferred response is obtained, it can be further employed for tasks of parameter estimation and to characterize the sensitivity of a sensing apparatus -- two aspects of central importance when devising a QS protocol (see Fig.~\ref{fig:1}). 

When inferring the value of an unknown parameter $\theta'$, we assume that one is given an estimate of the system response $\overline{\RC}(\theta')$, and the promise that $\theta'$ is sampled from  a known domain~$\Theta$. In such a case, one estimates the unknown parameter  as $\theta^*=\argmin_{\theta\in\Theta} |\widetilde{\RC}(\theta)-\overline{\RC}(\theta')|$. In many cases of interest, such as high-precision estimation of small magnetic fields, $\Theta$ will be small enough such that $\widetilde{\RC}(\theta)$ is bijective, ensuring that the solution  $\theta^*$ is unique. The resulting error in the estimate of the parameter $\theta'$ can be analyzed via the following corollary.
\begin{corollary}\label{cor:error_theta}
Let $\epsilon'$ be the estimation error in $\overline{\RC}(\theta')$ for some $\theta'$ in a known domain $\Theta$ where the  system response is bijective. Let $\chi$ be the  error introduced when estimating $\theta'$ via $\widetilde{\RC}(\theta)$ relative to the case when the exact response $\RC(\theta)$ is used. 
The number of shots, $N$, necessary to ensure that with a (constant) high probability $\chi$ is no greater than $\delta'$  is  $\Omega(\log^{3}(n)/(\delta' +\varepsilon')^2)$.
\end{corollary}

Corollary~\ref{cor:error_theta} certifies that $\widetilde{\mathcal{R}}(\theta)$ can be used to infer an unknown parameter from a measured system response without incurring additional uncertainties as long as enough shots are used. In fact, for fixed $\delta'$ and $\varepsilon'$, one only needs a poly-logarithmic number of shots.

Our inference-based method also allows for estimating  the sensitivity of QS schemes. Knowing the functional form of the response enables one to directly compute the sensitivity using the error propagation formula $(\Delta\theta)^2=(\Delta\mathcal{R}(\theta))^2 /|\partial_{\theta} \RC(\theta)|^2$~\cite{pezze2018quantum,sidhu2020geometric} which relates the variance $(\Delta \theta)^2$ in estimates of the parameter $\theta$ to the variance $(\Delta\RC(\theta))^2$ of the observable $O$ used to estimate $\theta$ (i.e., $(\Delta\RC(\theta))^2=\Tr[\DC\circ\SC_{\theta}\circ\EC(\rho_{\text{in}})O^2] - \Tr[\DC\circ\SC_{\theta}\circ\EC(\rho_{\text{in}})O]^2$) and to the slope $\partial_{\theta} \RC(\theta)$ of the response with respect to $\theta$. When $O^2=\id$ (i.e., measuring a Pauli operator), the sensitivity is
\begin{align}\label{eq:sensitivity}
    (\Delta\theta)^2=\frac{1-\left(\sum_{s=1}^n \left[a_{s} \cos(s\theta) + b_{s} \sin(s\theta)\right] + c\right)^2}{\vert \sum_{s=1}^n s\left( -a_{s} \sin(s\theta) + b_{s}  \cos(s\theta)\right) \vert^{2}}\,.
\end{align}
A similar expression will hold when using $\widetilde{\RC}(\theta)$ in place of $\RC(\theta)$. As shown in the SM, using $\widetilde{\RC}(\theta)$ to estimate the sensitivity at a parameter $\theta_l$ leads to an error which scales as $\OC(\varepsilon\log(n)/D_l)$, where $D_l=\partial_{\theta} \RC(\theta)\vert_{\theta=\theta_l}$. Moreover, as proved in the SM, a polynomial number of shots suffices to guarantee $|\Delta\theta-\Delta\widetilde{\theta}|\leq \delta ''$ for some fixed $\delta''$ if $D_l\in\Omega(1/\poly(n))$. Notably, the inferred sensitivity $(\Delta \theta)^2$ in Eq.~\eqref{eq:sensitivity}
can be compared with the quantum Cramer Rao-Bounds (CRBs)~\cite{hayashi2016quantum,liu2016quantum}, or the ultimate Heisenberg limit,  to determine the optimality of the sensing scheme.  In the SM we use this insight to show how our inferred response function can be used to train a measurement operator to reach the optimal sensing scheme given a fixed probe state. 


One can further ask whether Eq.~\eqref{eq:trig_response} can still be used in scenarios where the system response is no longer a  trigonometric polynomial.  Such a case will arise, for instance, if $\SC_\theta$ is not of the form in Eq.~\eqref{eq:unitary_family}. Still, we can leverage tools from trigonometric interpolation to accurately approximate the system response. Here, the following theorem holds for periodic responses and for parameters  close enough to the $\theta_k$ values in $P$   (regions of great interest for several QS tasks such as small magnetic field estimation).
\begin{theorem}\label{theo:bound2}
Let $f(\theta)$ be a $2\pi$-periodic function with $|f(\theta)|\leq 1$ $\forall \theta$, and let  $\widetilde{\RC}(\theta)$ be its trigonometric polynomial  approximation obtained from the $N$-shot estimates of $\overline{f}(\theta_k)$, with $\theta_k$ given by Eq.~\eqref{eq:sampling}. Defining the maximum estimation error $\varepsilon'=\max_{\theta_k\in P}|f(\theta_k)-\overline{f}(\theta_k)|$,  and assuming that  $|\theta-\theta_k| \in\OC(1/\poly(n))$, then
\begin{eqnarray}
    \vert f(\theta)-\widetilde{\mathcal{R}}(\theta) \vert \in\OC\left(\max\left\{\frac{M}{\poly(n)},\varepsilon'\log(n)\right\}\right)\, ,
\end{eqnarray}
where $M$ is the Lipschitz constant of $f(\theta)$.
\end{theorem}

Theorem~\ref{theo:bound2} shows that if $M\in\OC(n)$, which can occur for a wide range of parameter encoding schemes~\cite{sweke2020stochastic}, then, we can derive a result similar to that  in Corollary~\ref{cor:error}. Namely,  using a poly-logarithmic number of shots to estimate the quantities $\overline{f}(\theta_k)$  leads to $\widetilde{\mathcal{R}}(\theta)$ being a good approximation of  $f(\theta)$.

\medskip

\textit{Experimental results.} We demonstrate the performance of the inference method for a magnetometry task performed on the \textit{IBM\_Montreal} quantum computer. This consists in preparing the GHZ state, encoding a magnetic field  via Eq.~\eqref{eq:unitary_family} with $H =\sum_{j=1}^n Z_{j}$,  and measuring the parity operator $O=\bigotimes_{j=1}^nX_j$. Here, $Z_{j}$ and $X_j$ are  the Pauli $z$ and $x$ operators acting on the $j$-th qubit, respectively. We set $\DC$ to be  the identity channel and perform the QS task for systems of up to  $n=22$ qubits.

We first measure the system response at $2n+1$ training fields $\theta_k\in P$, sampled according to Eq.~\eqref{eq:sampling}. These estimates are then used to infer the response $\widetilde{\RC}(\theta)$ of Eq.~\eqref{eq:trig_response}, as well as to fit a function $g(\theta)=\alpha\cos(\beta \theta+ \gamma)+\zeta$. 
As discussed in the SM, the latter corresponds to a first order approximation of a noisy response under where the coefficients $\alpha$, $\beta$, $\gamma$ and $\zeta$, correct the cosine to account for the effects of hardware noise. To evaluate the ability of these two functions to recover the true response of the system, we compare their predictions against the measured system response at a set of random test fields.

In Fig.~\ref{fig:real_device_example_error_scaling}(a) 
we display inference results for $n=8$ and $n=16$ qubits, 
indicating that our method (red solid curve) is clearly able to fit the training and test fields better than the cosine response (black dotted curve). 
More quantitatively, in Fig.~\ref{fig:real_device_example_error_scaling}(b), we show the scaling of the error as a function of the system size. One can see that for all problem sizes considered  our method leads to smaller response prediction error. We note that for larger $n$ the effect of noise becomes more prominent, as the hardware noise suppresses the measured expectation values~\cite{wang2020noise,franca2020limitations,wang2021can}. Hence, in this regime both methods are equally limited by finite sampling noise which becomes of the same order as the magnitude of the response. Still, even for system sizes as large as $n=22$ qubits, the inference method reduces the relative error by a factor larger than two when compared to that of the $g(\theta)$ fit. Finally, we also use $\widetilde{\RC}(\theta)$ and $g(\theta)$ for parameter estimation, i.e., to determine an unknown  magnetic field encoded in the quantum state.  As shown in Fig.~\ref{fig:real_device_example_error_scaling}(c), the $g(\theta)$ fit matches the worst possible prediction already for $n=8$ qubits, whereas  our inference method can outperform the $g(\theta)$ fit by up to one order of magnitude.  In the SM we further discuss the behaviour of the parameter prediction curves of Fig.~\ref{fig:real_device_example_error_scaling}(c).

\begin{figure}[t]
    \includegraphics[width = \columnwidth]{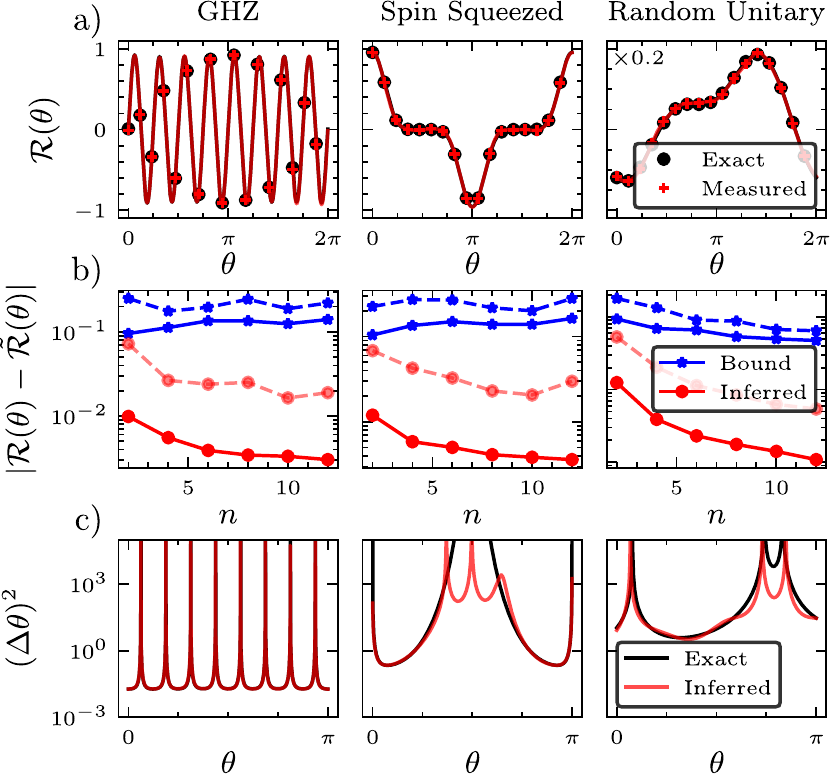}
    \caption{\textbf{Numerical results for QS tasks on a simulated noisy trapped ion device.} a) System response versus $\theta$ for $n=8$ qubits in  all three QS setups described in the main text. The exact response $\RC(\theta)$ (black curve), and its value at the training fields $\RC(\theta_k)$ (black points), were obtained with no finite sampling. In contrast, the response estimated at the training fields $\overline{\RC}(\theta_k)$ (red crosses), and the resulting inferred function $\widetilde{\RC}(\theta)$ (red curve), were obtained with a poly-logarithmic number of shots. b) Median (solid) and maximum (dashed) of the error $\vert \RC(\theta)-\widetilde{\RC}(\theta)\vert$ (red) and the bound of Theorem~\ref{theo:bound1} (blue) for $10^4$ test fields uniformly sampled over $(0,2\pi)$. The statistics were obtained over $30$ repetitions of the experimental setups. c) Curves depict the inferred sensitivity versus $\theta$. The exact (inferred) sensitivities are shown as a black (red) curves.}
    \label{fig:estimating_the_response}
\end{figure}

\medskip

\textit{Numerical simulations.} We complement the previous study with numerical results from a density matrix simulator that includes hardware noise, but where finite sampling can be omitted. We evaluate our inference method by  emulating several magnetometry tasks as they would have been performed on a trapped-ion quantum computer (see~\cite{cincio2018learning, trout2018simulating}). To this end, we consider three different sensing setups. First, we study the same standard GHZ magnetometry setting as implemented on the IBM device. Second, we characterize the squeezing in a system where the probe state is a spin coherent state,  $H =\sum_{j<k} X_{j} X_{k}$ is the  one-axis twisting Hamiltonian \cite{kitagawa1993squeezed}, and $O=Z_n$ (note that we did not chose the optimal measurement operator $O=\sum_j Z_j$ as we want to showcase that we can infer the response for any choice of $O$). Finally, we study a scenario where the probe state is constructed by a unitary composed of $4$ layers of a hardware efficient ansatz with random parameters~\cite{kandala2017hardware,cerezo2020cost},  $H =\sum_{j=1}^{n-1} Z_{j}Z_{j+1}$ and $O=\frac{1}{n}\sum_{i=1}^n X_i$. (This case is relevant for variational quantum metrology~\cite{cerezo2020variationalreview,meyer2020variational, koczor2020variational,beckey2020variational,kaubruegger2021quantum}, where one wishes to prepare a probe state via some parameterized quantum circuit that is usually initialized with random parameters.)  In all cases $\DC$ is the identity channel. See SM for further details, including the circuits employed.

As motivated by Corollary~\ref{cor:error}, $\widetilde{\RC}(\theta)$ is inferred with  $N=\lceil 5 \times 10^{2}\log(n)^{2}\log(2 \times 10^{2} (2n+1)) \rceil$ shots per $\theta_k$. Figure~\ref{fig:estimating_the_response}(a) shows that in all three QS settings considered the inferred response closely matches the exact one (i.e, the red curve for $\widetilde{\RC}(\theta)$ and the black curve for $\RC(\theta)$ are  overlaid). In  Fig.~\ref{fig:estimating_the_response}(b) we further show the scaling of the error $\vert \RC(\theta)-\widetilde{\RC}(\theta)\vert$  with respect to  the system size. This analysis reveals that our method always performs significantly better than the  upper bound given by Theorem~\ref{theo:bound1}. Indeed, we can  see that allocating a number of shots $N$ that increases poly-logarithmically  with $n$ allows the error  to decrease with increasing system size.  

Finally,  we use $\widetilde{\RC}(\theta)$ to  estimate the sensitivity  of  the three experimental setups. As shown in  Fig.~\ref{fig:estimating_the_response}(c), our method (red curves) recovers the  behavior of the exact sensitivity (black curves). The sensitivity diverges in parameter regions where  the experimental setup is insensitive  to the field (when the response function has a vanishing gradient). In the SM we further provide a theoretical and numerical analysis  for the estimated sensitivity,  as well as the scaling of the error of inferring an unknown parameter.

\textit{Conclusions.} We introduced a  inference-based scheme for QS which fully characterizes the response $\mathcal{R}(\theta)$ for a general class of unitary families by only measuring the system  at $2n+1$ known parameters.
This framework leverages techniques from quantum machine learning and polynomial interpolation~\cite{nakanishi2020sequential, dimatteo2022quantum,wierichs2022general} for quantum sensing, leading to new insights and methodology for the characterization, implementation and benchmarking of sensing protocols.

One of the main advantages of our  method is that it can be readily combined with existing sensing protocols. For instance, further research could explore the use of the inferred response function in a variational setting, involving an optimization of the experimental setup to maximize the sensitivity and parameter prediction accuracy (see SM). This paves the way for a new approach in data-driven quantum machine learning for QS where the optimization procedure does not require knowledge of the classical or quantum Fisher information \cite{koczor2020variational, beckey2020variational, cerezo2021sub, sone2020generalized,kaubruegger2021quantum,meyer2020variational,meyer2021fisher,cong2019quantum,pesah2020absence,sharma2020trainability,marshall2020characterizing,xue2021effects}.

\section*{Acknowledgments}

We thank Zoe Holmes, Michael Martin and Michael McKerns for helpful and insightful discussions. CHA and MC acknowledge support by NSEC Quantum Sensing at Los Alamos National Laboratory (LANL). MHG was supported by the U.S.
Department of Energy (DOE), Office of Science, Office of Advanced Scientific Computing Research, under the Quantum Computing Application Teams (QCAT) program.  AS was supported by the internal R\&D from Aliro Technologies, Inc. ATS, PJC, and MC were initially supported by the LANL ASC Beyond Moore's Law project.  MC was supported by the LDRD program of LANL under project number 20210116DR.  This work was also supported by the Quantum Science Center (QSC), a National Quantum Information Science Research Center of the U.S. Department of Energy (DOE). This research used quantum computing resources provided by the LANL Institutional Computing Program, which is supported by the U.S. DOE National Nuclear Security Administration under Contract No. 89233218CNA000001.

\bibliography{quantum.bib}

\clearpage
\newpage
\pagebreak[1]
\onecolumngrid
\section*{Supplemental Material for ``Inference-Based Quantum Sensing''}

\setcounter{theorem}{0}
\setcounter{corollary}{0}


In this Supplemental Material, we provide additional details for the manuscript ``\textit{Inference-Based Quantum Sensing}''. First, in section~\ref{app:Framework} we summarize the core idea behind our Inference method. Then, in Section~\ref{app:theo_response}, we present a proof for Theorem \ref{theo:trig-pol}. In Section~\ref{app:LSP}, we discuss the linear system problem that one needs to solve to determine the coefficients of the inferred system response. Therein we derive an optimal parameter sampling strategy that minimizes errors due to matrix inversion. We then present proofs of Theorem \ref{theo:bound1}, Corollary \ref{app:col1}, Corollary \ref{app:col2} and Theorem \ref{theo:bound2}, in Sections \ref{app:theo_bound1}, \ref{app:cor_error}, \ref{app:cor_error_theta} and \ref{app:theo_bound2} respectively. In Section~\ref{app:error_sensitivity} we derive and probe numerically Theorem \ref{th:error-sens} and Corollary \ref{app:col3} for the error in the sensitivity. This is followed by a first-order noise analysis for a GHZ magnetometry task (see Section \ref{app:first-order_noise}), numerical results further exploring the scaling of the errors arising in tasks of parameter estimation (see Section \ref{app:param_estimate}), and an application of our inference method in the context of variational quantum metrology (see Section \ref{app:training}). Finally, we present  examples of the circuits used for the experiments and numerical simulations performed (see Section~\ref{app:circuits}).

\section{Framework}\label{app:Framework}
We recall from the main text that in this work we consider a single-parameter QS setting employing an $n$-qubit probe state $\rho$ to estimate a parameter $\theta$. Here, $\rho$ is prepared by sending a fiduciary state $\rho_{\tin}$ through a state preparation channel $\EC$ such  that $\EC(\rho_\tin)=\rho$. Moreover, we first focus on the case of unitary families
\begin{equation}\label{eq:unitary_family}
   \mathcal{S}_\theta(\rho) = e^{-i \theta H/2}\rho ~e^{i \theta H/2}= \rho_{\theta}\,.
\end{equation}
Here, $H$ is a Hermitian operator such that $H=\sum_j h_j$ with $h_j^2=\id$, and $[h_j,h_{j'}]=0$, $\forall j, j'$. In addition, we allow for the possibility of sending $\rho_\theta$ through a second pre-measurement channel $\DC$ which outputs an $m$-qubit state $\DC(\rho_\theta)$ (with $m\leq n$), over which we measure the expectation value of an observable $O$, with $\norm{O}_\infty\leq 1$. The system response is thus defined as 
\begin{align}\label{eq:response}
    \RC(\theta)&=\Tr[\DC\circ\SC_{\theta}\circ\EC(\rho_\tin)O]\,.
\end{align}

the optimal strategy is to uniformly sample the parameters  as
\begin{equation}\label{eq:sampling}
   \theta_k=\frac{2 \pi (k-1)}{2n+1}\,,\quad \text{with $k=1,\ldots,2n+1$}\,,
\end{equation}

\section{Proof of Theorem \ref{theo:trig-pol}}\label{app:theo_response}

Here we present a proof for Theorem~\ref{theo:trig-pol}. This leverages the tools for optimizing quantum machine learning architectures presented in  Ref.~\cite{nakanishi2020sequential}.
\medskip
\begin{theorem}\label{theo:trig-pol}
Let $\RC(\theta) $ be the response function in Eq.~\eqref{eq:response} for a unitary family as in Eq.~\eqref{eq:unitary_family}. Then, for any $\EC$, $\DC$ and measurement operator $O$, $\RC(\theta) $ can be exactly expressed as a trigonometric polynomial of degree $n$. That is, 
\begin{align}\label{eq:trig_response}
    \RC(\theta) &= \sum_{s=1}^n \left[a_{s} \cos(s\theta) + b_{s} \sin(s\theta)\right] + c\,,
\end{align} 
with $\{a_s,b_s\}_{s=1}^n$ and $c$ being real valued coefficients.
\end{theorem}
\medskip
\begin{proof} 
We aim at deriving a closed-form expression for the true system response $\RC(\theta)$ for a unitary family $\SC_{\theta}$. 
Recall that, for an operator $h_j$ satisfying $h_j^2=\id$, we have
\begin{equation}
    e^{-i \theta h_j/2}=\cos\left(\frac{\theta}{2}\right)\id_j-i \sin\left(\frac{\theta}{2}\right) h_j\,.
\end{equation}
Hence, the state obtained after encoding by the unitary channel defined in Eq.~\eqref{eq:unitary_family} can be rewritten as
\begin{align}\label{eq:exact_state_expanded}
    \rho_{\theta} = \sum_{\vec{l},\vec{l}'} c_{\vec{l},\vec{l}'}(\theta) \left((-i)^{\norm{\vec{l}}^2} (i)^{\norm{\vec{l}'}^2} \vec{h}^{\vec{l}} \rho \vec{h}^{\vec{l}'}\right)\,,
\end{align}
with
\begin{align} \label{eq:coeff_cll} \nonumber
    c_{\vec{l},\vec{l}'}(\theta)&=\cos\left(\frac{\theta}{2}\right)^{2n-\norm{\vec{l}}^2-\norm{\vec{l}'}^2}\sin\left(\frac{\theta}{2}\right)^{\norm{\vec{l}}^2+\norm{\vec{l}'}^2} \\ 
        &= \sum_{s=1}^n \overline{a}_{\vec{l},\vec{l'},s}\cos(s\theta) + \sum_{s=1}^n \overline{b}_{\vec{l},\vec{l'},s}\sin(s\theta) + \overline{c}\,,
\end{align}
where, $\vec{l}=(l_1,l_2,\ldots,l_n)$, with $l_j\in\{0,1\}$, is a bitstring of length $n$, and where $\vec{h}^{\vec{l}}$ is a Hermitian operator defined as
\begin{equation}
    \vec{h}^{\vec{l}}:=h_1^{l_1}\otimes h_2^{l_2}\otimes \cdots\otimes h_n^{l_n}\,.
\end{equation}
Here, we have used the notation $h_j^0=\id_j$ and recall that  $\norm{\vec{l}},\norm{\vec{l}'}\in\{0,1,\ldots,n\}$. It can be further verified that the coefficients $c_{\vec{l},\vec{l}'}(\theta)$ are real-valued and symmetric with respect to swapping $\vec{l}$ and $\vec{l}'$.

Finally, we can recast the system response $\RC(\theta) = \Tr[\DC \circ \SC_{\theta} \circ \EC(\rho_\tin) O] $ as the trigonometric polynomial 
\begin{align}\label{eq:SM-trig}
    \mathcal{R}(\theta) &=\sum_{s=1}^n \left[ a_{s} \cos(s\theta) + b_{s} \sin(s\theta) \right] + c ~,
\end{align}
where we have defined the coefficients
\begin{align} \label{eq:coefficients} \nonumber
    a_s &= \sum_{\vec{l},\vec{l}'} \overline{a}_{\vec{l},\vec{l'},s} d_{\vec{l},\vec{l}'}\,,\\ \nonumber
    b_{s} &= \sum_{\vec{l},\vec{l}'} \overline{b}_{\vec{l},\vec{l'},s} d_{\vec{l},\vec{l}'}\,,\\ \nonumber
    c &= \overline{c} \sum_{\vec{l},\vec{l}'}d_{\vec{l},\vec{l}'} \,. \\
    d_{\vec{l},\vec{l}'} &= \Tr\left[ \DC\circ\SC_{\theta} \left((-i)^{\norm{\vec{l}}^2} (i)^{\norm{\vec{l}'}^2} \left(\vec{h}^{\vec{l}} \rho \vec{h}^{\vec{l}'}\right) \right) O \right]~.
\end{align}
\end{proof}

One of the striking implications of this theorem is that it holds even in the presence of quantum noise.
If we assume that the quantum hardware employed is prone to noise, then such noise will affect the probe state preparation process resulting in a noisy channel $\tilde{\EC}$ acting on $\rho_\tin$. The result is a noisy state $\widetilde{\rho}=\tilde{\EC}(\rho_\tin)$. As such, we can see that the effect of a noisy state preparation channel is to replace $\rho$ by $\widetilde{\rho}$ in Eq.~\eqref{eq:exact_state_expanded}, which does not ultimately change the functional form  of the system's response  Eq.~\eqref{eq:SM-trig} (it only changes the value of the coefficients in Eq.~\eqref{eq:coefficients}). 
One can further assume that there is a $\theta$-independent noise channel acting after the $\SC_\theta$ as well as consider the case where the  pre-measurement channel is noisy.  Similarly to the previous case, this does not change the  form  of $\RC(\theta)$ in Eq.~\eqref{eq:SM-trig}. Hence, we can see that the action of the noise channels can be ultimately absorbed into the  coefficients  of Eq.~\eqref{eq:coefficients}, and thus, we can still use the tools from the main text to fully characterize the system response by  measuring $2n+1$.

\section{Linear system problem}\label{app:LSP}

\subsection{Matrix $A$ }\label{app:Matrix_A}

As outlined in the main text, $2n+1$ evaluations of the response function, in Eq. \eqref{eq:trig_response}, yield a system of $2n+1$ equations with $2n+1$ unknown variables. 
That is, we obtain a Linear System Problem (LSP) of the form $A \cdot \vec{x} = \vec{d}$, where, $\vec{x} =(a_1,\ldots a_{n}, b_1,\ldots b_{n}, c)$ is the vector of coefficients that characterize the physical process at hand, $\vec{d} = (\mathcal{R}(\theta_1), \mathcal{R}(\theta_2), \ldots \mathcal{R}(\theta_{2n+1}))$ is the vector of system responses for all the $\theta_k \in P$, and $A$ is a $(2n+1)\times (2n+1)$ matrix defined as
\begin{eqnarray}
\hspace{-0.3cm}A=\begin{pmatrix}
    \cos(\theta_1) & \ldots & \cos(n \theta_1) & \sin(\theta_1) & \ldots & \sin(n \theta_1) & 1\\
    \cos(\theta_2) & \ldots & \cos(n \theta_2) & \sin(\theta_2) & \ldots & \sin(n \theta_2) & 1\\
     & \ldots & & & \ddots & & \\
    \cos(\theta_{2n+1}) & \ldots & \cos(n \theta_{2n+1}) & \sin(\theta_{2n+1}) & \ldots & \sin(n \theta_{2n+1}) & 1
    \end{pmatrix}. \nonumber
\end{eqnarray}

The solution $\vec{x}=A^{-1}\cdot\vec{d}$ of the LSP provides a full characterization of $\RC(\theta)$. In practice, the entries of the vector $\vec{d}$ are noisy estimates with errors resulting from finite sampling. 
We denote by $\vec{d}'$ the vector of estimated system responses and by $\vec{\delta}$ the vector of errors associated, such that
\begin{eqnarray}
    \vec{d}'=\vec{d}+\vec{\delta}\,.
\end{eqnarray}
When solving the LSP, the error in the coefficients obtained is given by 
\begin{eqnarray*}
   \epsilon= \norm{\vec{x}-\vec{x}'}=\norm{A^{-1}\cdot\vec{\delta}}\,,
\end{eqnarray*}
which can be bounded as
\begin{align}
     \epsilon\leq \norm{A^{-1}}\cdot\norm{\vec{\delta}}\,.
\end{align}
Using the $2$-norm leads to $\norm{A^{-1}}_2=1/\sigma_{\min}$, where $\sigma_{\min}$ is the smallest singular value of $A$. 
Furthermore, we can bound the smallest singular value of an $M\times M$ square real-valued matrix~\cite{hong1992lower}, which for the matrix $A$ yields
\begin{align}
    \sigma_{\min}\geq\left(\frac{M-1}{M}\right)^{(M-1)/2}|\det(A)|\cdot\frac{r_{\min}(A)}{\prod_{j=1}^{2n+1}r_j(A)}\,,
\end{align}
where $r_j(A)$ is the $2$-norm of the $j$th row of $A$, and $r_{\min}(A)=\min\{r_1(A), r_2(A), \cdots, r_{2n+1}(A)\}$. Using the fact that 
\begin{align}
\left(\frac{M-1}{M}\right)^{(M-1)/2}\geq \frac{1}{\sqrt{e}}~,
\label{eq:1/e}
\end{align}
we find
\begin{align}
    \varepsilon\leq \frac{\sqrt{e}\norm{\vec{\delta}}}{r_{\min}(A)|\det(A)|}\prod_{j=1}^{2n+1}r_j(A)\,.\label{eq:error}
\end{align}
For the error in Eq.~\eqref{eq:error} to be minimized, one needs to pick the $2n+1$ field values such that $\det(A)$ is maximized. Using the fact that the matrix $A$ can be related to a Vandermonde matrix \cite{horn1991topics} via element permutation, 
\begin{align}\label{eq:det-A}
    |\det(A)|&=\frac{1}{2^n}\prod_{i<j}\left|e^{i \theta_i}-e^{i \theta_j}\right|.
\end{align}

\subsection{Optimal sampling strategy}\label{app:Sampling}
In order to reduce the error induced when solving the LSP, one needs to maximize the determinant of $A$ in Eq.~\eqref{eq:det-A} by choosing the parameters $\theta_k \in P$ accordingly.  
Taking the logarithm of both sides in  Eq.~\eqref{eq:det-A}, one can see that maximizing $|\det(A)|$ is equivalent to  maximizing the quantity:
\begin{align}
    F    =\sum_{i<j}\log(|z_i-z_j|)\,,
\end{align}
where $z_i=e^{i \theta_i}$. 
Then, our optimization problem is equivalent to the task of placing $2N+1$ points over the unit circle such that they maximize the sum of the log of their distances. 
Since there are different ways to define a distance, we can replace $F$ by a proxy quantity $\widetilde{F}=\sum_{i<j}D_{i,j}$, where $D_{i,j}$ is taken to be a faithful distance between points over a circle. For convenience, let us pick the arc-length squared, i.e., $D_{i,j}=(\theta_i-\theta_j)^2=\Delta \theta_{i,j}^2$ such that, we now need to maximize the function
\begin{align}
    \widetilde{F}=\sum_{i<j}\Delta \theta_{i,j}^2\,,
\end{align}
subject to the condition
\begin{align}
    G=\sum_{i}\Delta \theta_{i,i+1}-2\pi=0\,.
\end{align}
This can be turned into the Lagrange multiplier problem
\begin{align}
    \LC=F-\lambda G\,,
\end{align}
with partial derivatives
\begin{align}
\frac{\partial\LC}{\partial \Delta \theta_{i,j}} &=2\Delta \theta_{i,j}-\lambda\delta_{j,i+1}\label{eq:lag-mult1}\,,\\
\frac{\partial\LC}{\partial \lambda} &=\sum_{i}\Delta \theta_{i,i+1}-2\pi\,,\label{eq:lag-mult2}
\end{align}
with $\delta_{j,i+1}$ representing Kronecker's delta function. The maximum will arise when the partial derivatives are zero, leading to
\begin{equation}\label{eq:delta_theta}
    \Delta \theta_{i,i+1}=\frac{2\pi}{2n+1}\, ,
\end{equation}
resulting in values of $\theta_i$ evenly distributed between $0$ and $2\pi$. 

\section{Proof of Theorem \ref{theo:bound1}}\label{app:theo_bound1}
\begin{theorem} \label{theo:bound1}
Let $\RC(\theta)$ be the exact response function, and let $\widetilde{\RC}(\theta)$ be its approximation obtained from the $N$-shot estimates $\overline{\RC}(\theta_k)$ with $\theta_k$ given by Eq.~\eqref{eq:sampling}.  Defining the maximum estimation error $\varepsilon=\max_{\theta_k\in P}|\RC(\theta_k)-\overline{\RC}(\theta_k)|$, then we have that for all $\theta$
\begin{equation}
    \vert \RC(\theta)-\widetilde{\mathcal{R}}(\theta) \vert \in\OC\left(\varepsilon\log(n)\right)\,.
\end{equation}
\end{theorem} 
Before proving Theorem \ref{theo:bound1}, we recall the following lemma that upper bounds trigonometric interpolation errors \cite{jackson1913accuracy}.
\medskip
\begin{lemma} \label{lemma:Upper_bound_TI} 
Let $f(\theta)$ be a periodic function of $\theta$ with period $2\pi$. Then, one can approximate $f(\theta)$ by a trigonometric function of order $n$
\begin{equation*}
    S_{n}(\theta) = \sum_{s=1}^n \left[a_{s} \cos(s\theta) + b_{s} \sin(s\theta)\right] + c
\end{equation*}
where $c= \frac{a_{0}}{2}$,
\begin{eqnarray}\nonumber
    a_{s} &=& \frac{2}{2n+1} \sum_{k=1}^{2n+1} f(\theta_{k}) \cos (s \theta_{k}), \qquad \text{and} \qquad
    b_{s} = \frac{2}{2n+1} \sum_{k=1}^{2n+1} f(\theta_{k}) \sin (s \theta_{k}),
\end{eqnarray}
such that if $\vert f(\theta_k) \vert \leq f_{max} ~ \forall \theta_k$ then,
\begin{equation*}
    \vert S_{n}(\theta) \vert \leq 5 f_{max} \log (n).
\end{equation*}
\end{lemma}
A proof for Lemma \ref{lemma:Upper_bound_TI} can be found in Ref. \cite{jackson1913accuracy}, and we note that the coefficients $a_{s}, b_{s}$ defined above exactly match those recovered when solving the LSP previously formulated. We now provide a proof for Theorem \ref{theo:bound1}.
\begin{proof}
Let $\mathcal{R}(\theta)$ be the true response function, and let $\widetilde{\mathcal{R}}(\theta)$ be the inferred response obtained from the estimates $\overline{\RC}(\theta_k)$ of the system response at $2n+1$ fields $\theta_k$ defined according to Eq.~\eqref{eq:sampling}. Then, the difference between two trigonometric functions is another trigonometric function which is also $2\pi$ periodic, namely,
\begin{eqnarray}\nonumber
    \mathcal{R}^{\ast}(\theta) &=& \mathcal{R}(\theta)-\widetilde{\mathcal{R}}(\theta) \\
    &=& \sum_{s=1}^n \left[a_{s}^{\ast} \cos(s\theta) + b_{s}^{\ast} \sin(s\theta)\right] + c^{\ast}
\end{eqnarray}
with $a_{s}^{\ast}$, $b_{s}^{\ast}$, and $c^{\ast}$ real-valued coefficients such that
\begin{eqnarray}\nonumber
    a_{s}^{\ast} = \frac{2}{2n+1} \sum_{k=1}^{2n+1} \left(\RC(\theta_{k})-\overline{\RC}(\theta_k)\right) \cos (s \theta_{k}), \quad \text{and} \qquad
    b_{s}^{\ast} = \frac{2}{2n+1} \sum_{k=1}^{2n+1} \left(\RC(\theta_{k})-\overline{\RC}(\theta_k)\right) \sin (s \theta_{k}).
\end{eqnarray}
Hence, defining $\varepsilon=\max_{\theta_k\in P}|\RC(\theta_k)-\overline{\RC}(\theta_k)|$, and using Lemma \ref{lemma:Upper_bound_TI}, shows that 
\begin{equation}\label{eq:five_log}
    \vert \mathcal{R}^{\ast}(\theta) \vert=\vert\mathcal{R}(\theta)-\widetilde{\mathcal{R}}(\theta)\vert \leq  5\varepsilon \log (n).
\end{equation}
The previous shows that for all $\theta$, we have $\vert\mathcal{R}(\theta)-\widetilde{\mathcal{R}}(\theta)\vert \in \OC( \varepsilon \log (n))$.
\end{proof}
\section{Proof of Corollary \ref{app:col1}}\label{app:cor_error}
\begin{corollary}\label{app:col1}
The number of shots $N$, necessary to ensure that with a (constant) high probability, and for all $\theta$, the error $\vert \RC(\theta)-\widetilde{\mathcal{R}}(\theta) \vert\leq \delta$, for an inference  error $\delta$, is in $\Omega\left( \log^3(n)/\delta^2\right)$.
\end{corollary}

\begin{proof}Let $\zeta_{k} = \vert \RC(\theta_{k}) - \overline{\RC}(\theta_{k}) \vert $ be the estimation error in the system response for a parameter $\theta_{k}$. We start this proof by bounding the probability of such error to be higher than a threshold $\epsilon$.
Recall that an estimate $\overline{\RC}(\theta_{k})$ is obtained as an average over $N$ measurements of an observable $O$. We denote by $X^i_k$ (with $i=1,\hdots,N$) the random variable associated with each measurement. Given that we consider observables $O$ with norm $\norm{O}_\infty\leq 1$, the outcomes of such a measurement take values in the range $[a^i_k, b^i_k]$ who's amplitude can always be bounded as $b^i_k-a^i_k\leq2$ . Using this notation, $\overline{\RC}(\theta_{k})$ is defined as the average $\sum^N_{i=1} (X_k^i/N)$, which is an unbiased estimate of the response, i.e., $E[\overline{\RC}(\theta_{k})] = \RC(\theta_{k})$. 
It can be shown that:
\begin{eqnarray}\nonumber 
    \text{Pr}(\zeta_{k} = |\RC(\theta_{k}) - \overline{\RC}(\theta_{k})| \geq \epsilon) &\leq& 2\, \exp \bigg[-\frac{2\epsilon^2}{\sum^N_{i=1} [(b_k^i - a_k^i)/N]^2} \bigg] \\ 
    &\leq& 2\, \exp [-N\epsilon^2/2]\,,\label{eq:one_err}
\end{eqnarray}
where we have made use of Hoeffding's inequality to derive the first inequality, and of the fact that $b^i_k-a^i_k\leq 2$ to obtain the second one.

As we are interested in the \emph{maximum} estimation error $\epsilon$ appearing in Theorem~\ref{app:theo_bound1}, we now aim to bound the probability that \emph{any} of the $2n+1$ estimation errors $\zeta_{k}$ exceed $\epsilon$. That is, we wish to bound $\text{Pr}(\cup_{k=1}^{2n+1} E_{k})$, where $E_{k}$ is defined as the event when $\zeta_{k}$ is greater than $\epsilon$. 
This is readily achieved by means of Boole's inequality which states that the probability of at least one event (over a set of events) to occur is no greater than the sum of the probabilities of each individual event in the set. Hence we have
\begin{eqnarray} \nonumber 
    \text{Pr} \left( \cup_{k=1}^{2n+1} E_k \right) &\leq& \sum_{k=1}^{2n+1} \text{Pr}(E_{k}) \\ \nonumber
    &\leq&  \sum_{k=1}^{2n+1} 2\,\text{exp}[- N \epsilon^2/2]\\
    &=& (4n+2)\,\text{exp}[- N \epsilon^2/2]
\end{eqnarray}
where in the first inequality we used Boole's inequality, while in the second one we have used Eq.~\eqref{eq:one_err}.

Thus, to ensure that $\text{Pr} \left( \cup_{k=1}^{2n+1}E_{k}  \right) \leq a$ (for a some small positive $a$), we require a number of shots $N$,  such that
\begin{equation}\label{eq:Nmin}
    N \geq \frac{2\,\text{log}[(4n+2)/a]}{\epsilon^2}.
\end{equation}

Finally, we can relate estimation errors to inference errors according to Theorem~\ref{theo:bound1}. In particular it follows from Eq.~\eqref{eq:five_log} that $\text{Pr}(\vert \RC(\theta)-\widetilde{\mathcal{R}}(\theta) \vert \geq 5\epsilon \log (n)) \leq a$ for $N$ given in~\eqref{eq:Nmin}. 
Accordingly, to ensure that $\text{Pr}(\vert \RC(\theta)-\widetilde{\mathcal{R}}(\theta) \vert \geq \delta) \leq a$, with $\delta=5\epsilon \log (n)$, we require that
\begin{equation}\label{eq:shots}
    N\geq \frac{50\,\text{log}^{2}(n)\,\text{log}[(4n+2)/a]}{ \delta^2}.
\end{equation}
Therefore, with a number of shots $N\in\Omega( \log^{3}(n)/\delta^{2})$ we can ensure that $ \vert \RC(\theta) - \widetilde{\RC}(\theta) \vert <  \delta$ with constant probability $1-a$.
\end{proof}

\section{Proof of Corollary \ref{app:col2}}\label{app:cor_error_theta}
\begin{corollary}\label{app:col2}
Let $\epsilon'$ be the estimation error in $\overline{\RC}(\theta')$ for some $\theta'$ in a known domain $\Theta$ where the  system response is bijective. Let $\chi$ be the  error introduced when estimating $\theta'$ via $\widetilde{\RC}(\theta)$ relative to the case when the exact response $\RC(\theta)$ is used. 
The number of shots, $N$, necessary to ensure that with a (constant) high probability $\chi$ is no greater than $\delta'$  is  $\Omega(\log^{3}(n)/(\delta' +\varepsilon')^2)$.
\end{corollary}

\begin{proof}

Let us here consider the problem of estimating an unknown parameter $\theta' \in \Theta$ given a measured system response $\overline{\RC}(\theta')$ with estimation error $\varepsilon'$.  Ideally, one would like to use the exact system response $\RC(\theta)$ for such task. 
However, in practice, one does not have access to $\RC(\theta)$ but only to the inferred response $\widetilde{\RC}(\theta)$, and we now aim at bounding the effect introduced by such approximation.

As sketched on the main text, the error propagation formula allows us to relate the measured system response uncertainty $\Delta\RC(\theta)$ to the uncertainty  $\Delta \theta$ in the estimated parameter as $(\Delta\theta)^2=(\Delta\RC(\theta))^2 /|\partial_{\theta} \RC(\theta)|^2$. 
 From Theorem \ref{theo:bound1}, we know that using  $\widetilde{\RC}(\theta)$ for determining $\theta'$ introduces an error that is upper bounded by $5 \varepsilon \log(n)$ (where $\varepsilon$ is the maximum estimation error over the $2n+1$ parameters $\theta_k$ used to construct $\widetilde{\RC}(\theta)$). A simple geometric argument shows that the uncertainty in the error propagation formula now becomes  $\Delta\widetilde{\RC}(\theta) \leq 5 \varepsilon \log(n)+\varepsilon'$ (see Fig.~\ref{fig:uncertainty}). 
 
 \begin{figure}[h]
    \centering
    \includegraphics[width=1\linewidth]{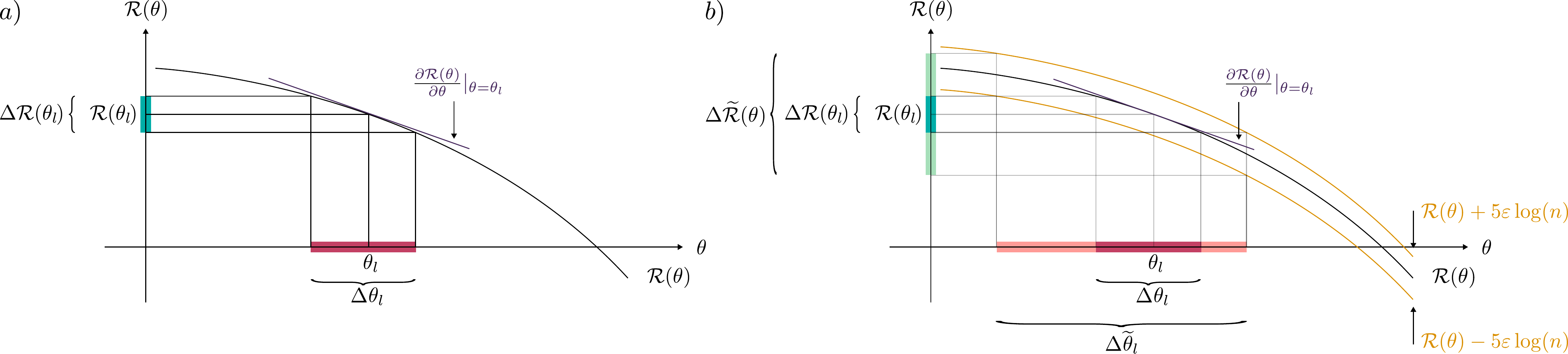}
    \caption{
    \textbf{Error from the inference method.} a) Given the true response, one can use the error propagation formula to relate the uncertainty in the response $\Delta \RC(\theta)$ to an uncertainty in the parameter $\Delta \theta$. b) when using the inferred response $\widetilde{\RC}(\theta)$ one can consider the worst-case scenario curves $\RC(\theta)\pm5\varepsilon\log(n)$ to obtain bounds on the uncertainty $\Delta \widetilde{\mathcal{R}}(\theta)$ of the approximate response and also on the ensuing uncertainty $\Delta \widetilde{\theta}$ in estimates of the parameter. Specifically, one can see that $\Delta\widetilde{\RC}(\theta) \leq \Delta\RC(\theta)+5 \varepsilon \log(n)$.}
    \label{fig:uncertainty}
\end{figure}
 
To ensure that the error introduced by the use of $\widetilde{\RC}(\theta)$ instead of $\RC(\theta)$ does not dominate, we require that $5 \varepsilon \log(n) - \varepsilon'  \geq \delta$, with $\delta$ a positive value, should occur with low probability. 
In other words, we wish that $\varepsilon   \geq (\delta' + \varepsilon')/5\log(n)$ only occurs with a low probability. From the proof of Corollary \ref{app:col1}, we know that
\begin{eqnarray}
    \text{Pr}(\varepsilon \geq \eta)\leq (4n+2)\,\text{exp}[- N \eta^2/2].
\end{eqnarray}
Setting $\eta = (\delta' + \varepsilon')/5\log(n)$ leads to
\begin{eqnarray}
    \text{Pr}(\varepsilon \geq \eta) &\leq& (4n+2)\,\text{exp}\bigg[- \frac{N(\delta' + \varepsilon')^2}{50\,\text{log}^2(n)}\bigg]
\end{eqnarray}
Therefore to ensure that $\text{Pr}(5 \varepsilon \log(n) - \varepsilon' \geq \delta' ) \leq b$ we require the number of shots to satisfy
\begin{equation}\label{eq:shots-needed}
    N \geq \frac{50\,\log^2(n) \log[(4n+2)/b]}{(\delta' + \varepsilon')^2},
\end{equation}
i.e., such that $N\in\Omega(\log^{3}(n)/(\delta' +\varepsilon')^2)$.

\end{proof}

\section{Proof of Theorem \ref{theo:bound2}}\label{app:theo_bound2}
\begin{theorem}\label{theo:bound2}
Let $f(\theta)$ be a $2\pi$-periodic function with $|f(\theta)|\leq 1$ $\forall \theta$, and let  $\widetilde{\RC}(\theta)$ be its trigonometric polynomial  approximation obtained from the $N$-shot estimates of $\overline{f}(\theta_k)$, with $\theta_k$ given by Eq.~\eqref{eq:sampling}. Defining the maximum estimation error $\varepsilon'=\max_{\theta_k\in P}|f(\theta_k)-\overline{f}(\theta_k)|$,  and assuming that  $|\theta-\theta_k| \in\OC(1/\poly(n))$, then
\begin{eqnarray}
    \vert f(\theta)-\widetilde{\mathcal{R}}(\theta) \vert \in\OC\left(\max\left\{\frac{M}{\poly(n)},\varepsilon'\log(n)\right\}\right)\, ,
\end{eqnarray}
where $M$ is the Lipschitz constant of $f(\theta)$.
\end{theorem}

Before providing a proof for Theorem \ref{theo:bound2}, we introduce the following lemma that bounds the error of using a trigonometric polynomial to approximate a  periodic function \cite{petras1995error}.
\medskip
\begin{lemma} \label{lemma:petras_error} 
Let $f(\theta)$ be a $2\pi$-periodic function with Lipschitz-constant $M$ and let $S_{n}$ represent the trigonometric interpolation polynomial of order $n$ used to approximate  $f$ via the values $f(\theta_k)$ at $2n+1 $ equidistant nodes, $\theta_{k}= \frac{2 k \pi}{2n+1}$, $k=0,1,\ldots,2n$. The trigonometric approximation error is such that for all $\theta$
\begin{equation*}
    \vert f(\theta) - S_{n}(\theta) \vert \in \OC\left(M \frac{\log (n)}{n} \left\vert \sin \big( (n+1/2)\,\theta \big) \right\vert\right).
\end{equation*}
\end{lemma}
A proof for Lemma \ref{lemma:petras_error} can be found on Ref. \cite{petras1995error}. We now provide a proof for Theorem \ref{theo:bound2}.
\begin{proof}
Let $f(\theta)$ be the response function of the system, 
and let $\mathcal{R}(\theta)$ be a trigonometric polynomial of order $n$ used to estimate $f(\theta)$. 
It follows from Lemma \ref{lemma:petras_error} that when $f$ is $2\pi$-periodic with Lipschitz constant $M$, we have 
\begin{equation}\label{eq:petras}
    \vert f(\theta) - \mathcal{R}(\theta) \vert \in \OC \left(M \frac{\log(n)}{ n} \left\vert \sin \big( (n+1/2)\,\theta \big) \right\vert\right).
\end{equation}
Furthermore, when $|\theta - \theta_{k}| \leq \OC\left(1/\poly(n)\right)$ for any $\theta_{k}=\frac{2\pi k}{2n+1}$, we can verify using the small angle approximation that  
\begin{eqnarray} \nonumber
   \left\vert \sin \big( (n+1/2)\,\theta \big) \right\vert \in \mathcal{O}(1/n).
\end{eqnarray}
Hence, Eq. \eqref{eq:petras} can be rewritten as
\begin{equation}
    \vert f(\theta) - \mathcal{R}(\theta) \vert \leq \mathcal{O}\left( \frac{M \log(n)}{n^{2}}\right).
\end{equation}

So far, we have only considered the error of approximating $f(\theta)$ by a trigonometric polynomial $\RC(\theta)$ obtained from the set of system responses $f(\theta_k)$, with  $k=0,\ldots,2n$.  However, in practice we can only estimate such responses via measurements on a quantum device. Defining  $\overline{f}(\theta_k)$ as the $N$-shot estimates of the system response, and defining $\widetilde{\RC}(\theta)$ as the trigonometric polynomial obtained from  the measured values $\overline{f}(\theta_k)$, we now aim at bounding the error $\vert f(\theta) - \widetilde{\mathcal{R}}(\theta) \vert$. To do so, one can use the following chain of inequalities
\begin{eqnarray*} 
    \vert f(\theta) - \widetilde{\mathcal{R}}(\theta) \vert & = & \vert f(\theta)- \mathcal{R}(\theta) + \mathcal{R}(\theta) - \widetilde{\mathcal{R}}(\theta) \vert \\ 
    &\leq & \vert f(\theta)- \mathcal{R}(\theta)\vert + \vert \mathcal{R}(\theta) - \widetilde{\mathcal{R}}(\theta) \vert \\ 
    &\leq& 2\max\left\{ \vert f(\theta)- \mathcal{R}(\theta)\vert, \vert \mathcal{R}(\theta) - \widetilde{\mathcal{R}}(\theta) \vert\right\}.
\end{eqnarray*}
Defining the maximum estimation error $\varepsilon'=\max_{\theta_k\in P}|f(\theta_k)-\overline{f}(\theta_k)|$  and leveraging results from Theorem \ref{theo:bound1} and Lemma \ref{lemma:petras_error}, and assuming that $|\theta-\theta_k| \in\OC(1/\poly(n))$, leads to
\begin{equation}
    \vert f(\theta) - \widetilde{\mathcal{R}}(\theta) \vert \in\OC\left(\max\left\{\frac{M}{\poly(n)},\varepsilon'\log(n)\right\}\right).
\end{equation}
\end{proof}

\section{Error in the sensitivity}\label{app:error_sensitivity}

As stated in the main text, once the response function is known, one can use it to directly compute the sensitivity at a field $\theta=\theta_l$ using the error propagation formula~\cite{pezze2018quantum}
\begin{equation}
    \Delta\theta_k=\frac{\Delta\RC(\theta_k) }{|\partial_{\theta} \RC(\theta)\vert_{\theta=\theta_k}\vert}\,.\label{eq:sens-ex}
\end{equation}
In practice, since  one does not have access to the exact response function, one needs to use its $N$-shot approximate $\widetilde{\RC}(\theta)$, which leads to an approximate sensitivity
\begin{equation}
\Delta\widetilde{\theta_l}=\frac{\Delta\widetilde{\RC}(\theta_l) }{|\partial_{\theta} \widetilde{\RC}(\theta)\vert_{\theta=\theta_l}\vert}\,.\label{eq:sens-ap}
\end{equation}
Here we will analyze the error induced by using $\widetilde{\RC}(\theta)$ instead of $\RC(\theta)$. Specifically, the following theorem holds. 
\begin{theorem}\label{th:error-sens}
Let $\RC(\theta)$ be the exact response function, and let $\widetilde{\RC}(\theta)$ be its approximation obtained from the $N$-shot estimates $\overline{\RC}(\theta_k)$ with $\theta_k$ given by Eq.~\eqref{eq:sampling}.  Defining the maximum estimation error $\varepsilon=\max_{\theta_k\in P}|\RC(\theta_k)-\overline{\RC}(\theta_k)|$, and the slope of $\RC(\theta)$ at a field $\theta_k$  $D_l=|\partial_{\theta} \widetilde{\RC}(\theta)\vert_{\theta=\theta_l}|$, then
\begin{equation}
    \vert \Delta\theta-\Delta\widetilde{\theta} \vert \in\OC\left(\frac{\varepsilon\log(n)}{D_l}\right)\,.
\end{equation}
\end{theorem}
\begin{proof}
First, using a geometric argument similar to that in the proof of  Corollary \ref{app:col2} in Section~\ref{app:cor_error_theta} we have that  $\Delta\widetilde{\RC}(\theta) \leq \Delta\RC(\theta)+ 5 \varepsilon \log(n)$. Then, we can see from Fig.\ref{fig:uncertainty}(b) that one can write 
\begin{equation}
\Delta\widetilde{\theta_l}\leq \frac{\Delta\RC(\theta)+ 5 \varepsilon \log(n)}{|\partial_{\theta} \RC(\theta)\vert_{\theta=\theta_l}\vert}\,.
\end{equation}
Since, by definition, we have that $\Delta\widetilde{\theta_l}\geq \Delta\theta_l$,
then the following inequality holds
\begin{align}
    |\Delta\widetilde{\theta_l}-\Delta\theta_l|&\leq \frac{5 \varepsilon \log(n)}{|\partial_{\theta} \RC(\theta)\vert_{\theta=\theta_l}\vert}\\
    &=\frac{5 \varepsilon \log(n)}{D_l}\,,
\end{align}
which completes the proof.

\end{proof}

Theorem~\ref{th:error-sens} indicates that using the inferred response to estimate the sensitivity leads to an error that scales linearly $\varepsilon$ and $\log(n)$, but inversely proportional to $D_l$. The previous is expected, as $\Delta\theta$ diverges in regions where the system response is flat, i.e., in regions where $D_l$ approaches zero. Notably, Theorem~\ref{th:error-sens}  allows us to derive the following corollary:
\medskip
\begin{corollary}\label{app:col3}
The number of shots $N$, necessary to ensure that with a (constant) high probability the error in the sensitivity $ \vert \Delta\theta-\Delta\widetilde{\theta} \vert\leq \delta''$ at parameter $\theta=\theta_l$, for an inference sensitivity   error $\delta''$, is in $\Omega\left(\frac{ \log^3(n)}{(D_l\delta'')^2}\right)$.
\end{corollary}
\begin{proof}
The proof follows by first considering that, with high probability, $\vert \RC(\theta)-\widetilde{\mathcal{R}}(\theta) \vert\leq \delta$ as in Corollary~\ref{app:col2}. The previous implies that,  with high probability, $   |\Delta\widetilde{\theta_l}-\Delta\theta_l|\leq \frac{\delta}{D_l}$ (one can geometrically see that this holds from Fig.~\ref{fig:uncertainty}(b)). Setting the equality $\frac{\delta}{D_l}=\delta''$ and replacing in 
Eq.~\eqref{eq:shots} proves that a number of shots $N\in\Omega( \log^{3}(n)/(D_l\delta'')^{2})$ are necessary to ensure that $ \vert \Delta\theta-\Delta\widetilde{\theta} \vert\leq \delta''$ with constant probability.


\end{proof}

In Fig.~\ref{fig:estimating_the_sensitivity}, we report a study of the scaling of the errors resulting from the estimation of the sensitivity of two magnetometry setups: the GHZ state with a parity measurement, and the spin squeezed state with a single qubit $Z$-measurement. Errors are assessed across a range where no divergences in the sensitivities are observed. 
For the GHZ state and the spin squeezed state, this occurs at the ranges given by $\theta \in (-\pi /3n, \pi /3n)$, and $\theta\in (\pi /3n, \pi/n )$, respectively. No plot for the setup involving random unitary preparation (discussed in the main text) is shown here as the inherent randomness in the preparation of the probe state makes it difficult to systematically identify a region where the sensitivity is well-behaved over varied system sizes.
\begin{figure}[h]
    \centering
    \includegraphics[scale = 1]{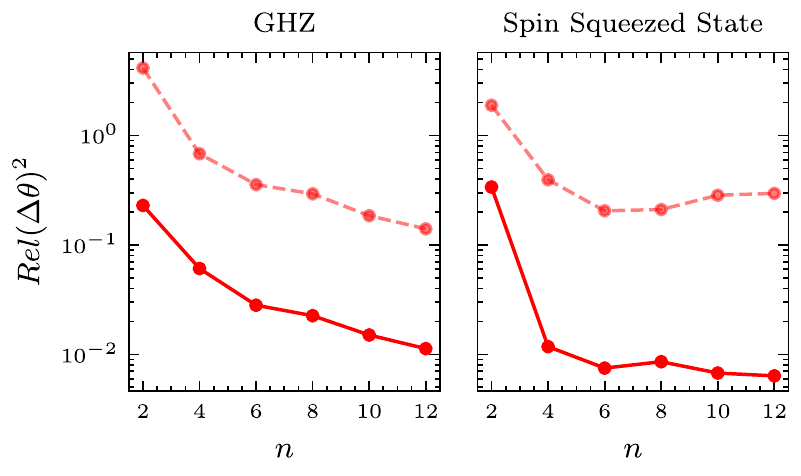}
    \caption{\textbf{Scaling of the sensitivity prediction errors.}
    We show the scaling of the median (solid) and maximal (dashed) relative errors in the sensitivity estimates for the GHZ (left panel) and spin squeezed (right panel) experimental setups. 
    The inferred response functions, used to predict the sensitivities, were estimated using $N=\lceil 5 \times 10^{2}\log(n)^{2}\log(2 \times 10^{2} (2n+1)) \rceil$.    The relative error is defined as the absolute error normalised by the range of sensitivity values. 
}
    \label{fig:estimating_the_sensitivity}
\end{figure}

\section{First-order noisy response for a GHZ magnetometry task }\label{app:first-order_noise}

Here we consider a magnetometry task  where one prepares the GHZ state, encodes a magnetic field  via 
\begin{equation}
   \mathcal{S}_\theta(\rho) = e^{-i \theta H/2}\rho ~e^{i \theta H/2}= \rho_{\theta}\,.
\end{equation}
with $H =\sum_{j=1}^n Z_{j}$,  and measures the parity operator $O=\bigotimes_{j=1}^nX_j$. Here, $Z_{j}$ and $X_j$ are  the Pauli $z$ and $x$ operators acting on the $j$-th qubit, respectively. Moreover, we assume that noise acts during the state preparation, parameter encoding, pre-measurement channel and measurement. 

First, let us consider the action of noise acting during the state preparation and pre-measurement channel measurement. As shown in~\cite{marshall2020characterizing,xue2021effects} to first order in the noise parameters, the noisy response function $\RC_{\text noisy}(\theta)$ can be expressed as
\begin{align}
    \RC_{\text noisy}(\theta)=a\RC_{\text noiseless}(\theta)+b\,,
\end{align}
where $\RC_{\text noiseless}(\theta)$ denotes the noiseless response function, while $a$ and $b$ are noise-dependent parameters. Then, assuming a coherent error during the parameter encoding, such that the source encodes a parameter $\theta+\delta\theta$ rather than $\theta$ one would get
\begin{align}
    \RC_{\text noisy}(\theta)=a\RC_{\text noiseless}(\theta+\delta\theta)+b\,.
\end{align}
For the special case of a magnetometry task, where the noiseless response is a simple cosine, the previous equations becomes 
\begin{align}\label{eq:noisy}
    \RC_{\text noisy}(\theta)=a\cos(\theta+\delta\theta)b\,.
\end{align}
Hence, assuming that Eq.~\eqref{eq:noisy} holds, it is reasonable to fit the noisy response function via $g(\theta)= \alpha\cos(\beta \theta+ \gamma)+\zeta$. We note however, that as shown in the numerics of the main text, such first order approximation does not hold when realistic noise models are considered, as the fit obtained through $g(\theta)$  can greatly differ from the measured response function.


\section{Error in the parameter estimation.}\label{app:param_estimate}
Once the inferred form $\widetilde{\RC}(\theta)$ of the response obtained, it can be used to estimate the value of a unknown field $\theta'$ in a given range. Given a measurement $\overline{\RC}(\theta')$ of the response for this field, and provided that the response is bijective in the range considered,
we define the estimated field value $\theta^{*}$ as follows:
\begin{equation}
    \theta^*=\argmin_\Theta |\widetilde{\RC}(\theta)-\overline{\RC}(\theta')|.
\end{equation}

This method was also used in the main text for the experimental data obtained from the IBM quantum computers. In Fig.~\ref{fig:error_in_the_inversion}, we report a study of the absolute error in such predictions for a GHZ setting numerically simulated over system sizes of up to $n=12$ qubits. The errors resulting from the use of the inferred response function (red curve) are compared to the errors resulting from the use of a  fit given by $g(\theta) = \alpha \cos(\beta\theta + \gamma) + \zeta$ (black curve). We find that the inferred function significantly outperforms the performance of the  fit given by $g(\theta)$.  Note that the inversion is estimated over the range $(\theta - \pi / 10n, \theta + \pi /10n)$ leading to a maximal error of $\pi/10n$. 

\begin{figure}[t]
    \centering
    \includegraphics[scale=1]{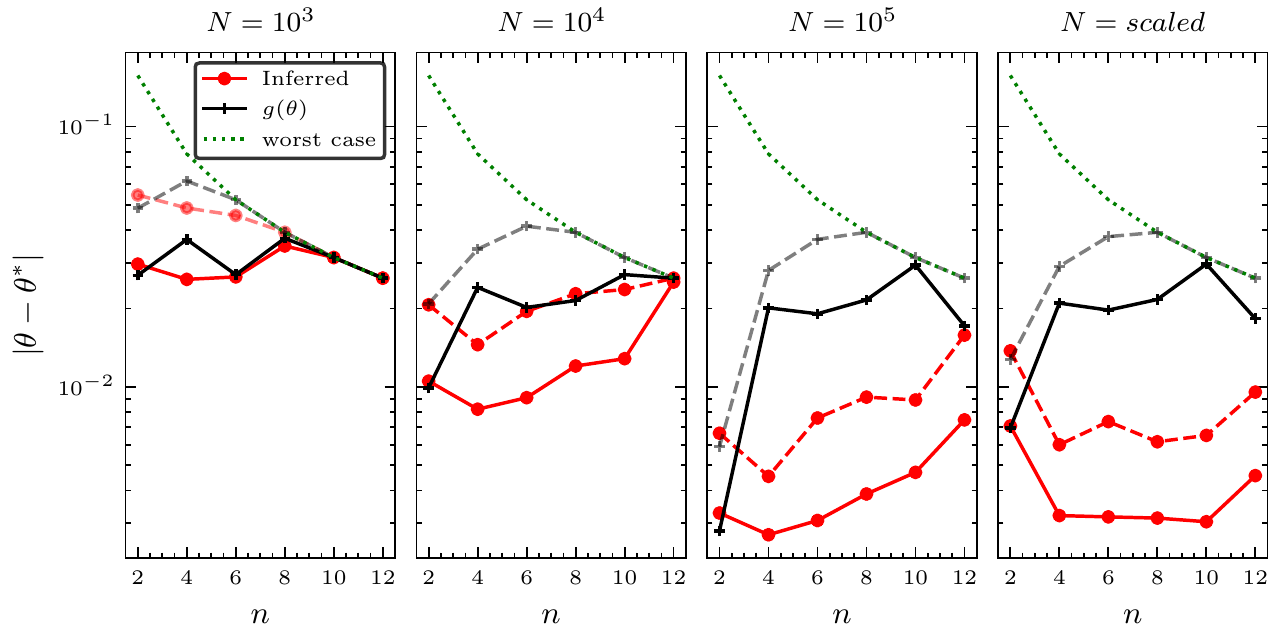}
    \caption{
    \textbf{Field prediction errors}.  We show the scaling of the median (solid) and upper quartile (dashed) absolute errors in the $\theta^*$ values predicted by inverting either the inferred response function (red curve) or a $g(\theta)$ fit  (black curve) for different numbers of shots.
    These are obtained for a magnetometry setting with a probe GHZ state, over $30$ repeats per system sizes $n\in[2,12]$, and with $30$ randomly selected $\theta'$ angles in the range $(0, 2\pi)$ per system size. Each repeat includes a different shot noise realization of the exact observable. The response functions were estimated using different numbers of shots with the $N=scaled$ label referring to the number of shots scaling as $N=\lceil 5 \times 10^{3}\log(n)^{2}\log(2 \times 10^{2} (2n+1)) \rceil$. }
    \label{fig:error_in_the_inversion}
\end{figure}

There are several competing effects at play when inverting the response function to estimate the value of a particular field. Below we describe all of them.

\begin{itemize}
    \item \textit{The effect of noise with increasing system size}:
\end{itemize}

\begin{figure}[h!]
    \centering
    \includegraphics[scale=1]{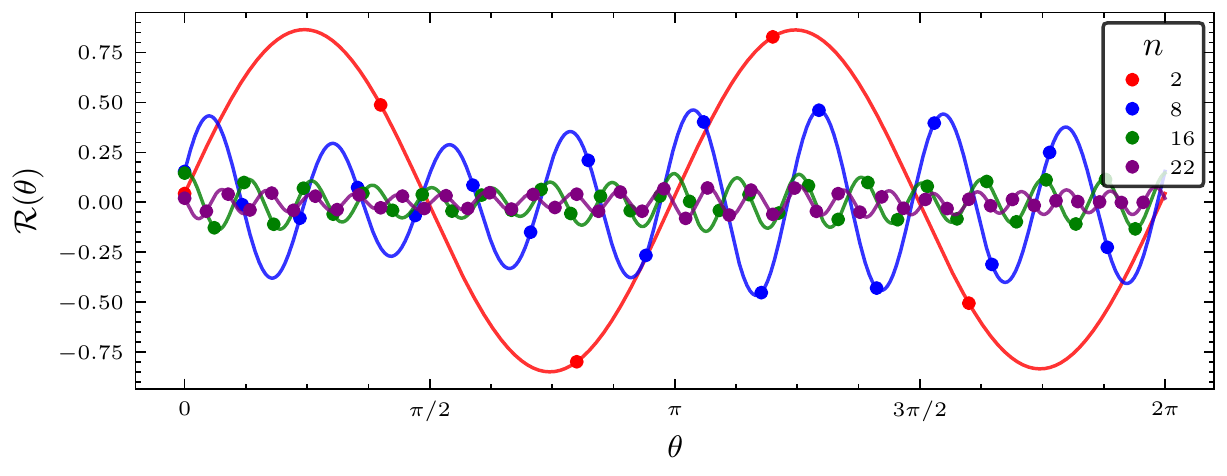}
    \caption{
    \textbf{Response functions obtained from the magnetometry task on IBM hardware}. This shows the inferred response functions obtained from the \textit{IBM\_Montreal} quantum computer. The number of shots used was $3.2 \times 10^{4}$ for each data point. }
    \label{fig:IBM_exp_data_same_axis}
\end{figure}

Firstly, as the system size increases the effect of noise becomes more significant, which in turn reduces the performance of the probe state (and overall sensing scheme) when estimating the field parameter. The intuition for this can be taken from Ref.~\cite{wang2020noise,franca2020limitations}, which shows that the measured signal from a noisy quantum computer decays exponentially with the number of layers, or depth, of the quantum circuit. In our case we are using $\mathcal{O}(n)$ layers to prepare our GHZ probe state. This results in a exponentially fast suppression of the response function with the number of qubits, resulting in a significant degradation of the probe state performance. This suppression can be observed in Fig. \ref{fig:IBM_exp_data_same_axis} which shows the amplitudes of the response functions as the system size is varied. Here one can verify that increased problem size, implies deeper circuits, and thus flatter response functions leading to a larger error in inversion. We do highlight, however, that  a larger estimation error is not due to the inference method we present here (whose goal is to characterize the response function), but rather a feature of the sensing set-up being noisy and thus sub-optimal. 


\begin{itemize}
    \item \textit{Finite sampling effects}:
\end{itemize}

In order to compensate for the noise-induced degradation of the sensing scheme sensitivity as the system size is increased, one needs to increase the number of shots used when determining the response function. As verified in noisy numerical simulations presented in Fig.~\ref{fig:error_in_the_inversion}, increasing the number of shots systematically improves the estimation errors yielded by the inferred response but not the errors corresponding to the cosine fit, as the latter is intrinsically limited in its ability to capture the noisy behaviour of the experimental response. This is because $g(\theta)$ may be an ill-approximation to how the noise truly acts, which cannot be solved by simply increasing the number of shots. In order to observe this effect for the small system sizes probed, the noise rates in our ion trap emulation were increased by a factor of $20$. 

\begin{itemize}
    \item \textit{The details of the inversion task}:
\end{itemize}

As previously mentioned, the inversion was calculated in a range taken to be $\theta \pm \pi/10n$. Therefore, the maximum possible error for each predicted field value is the extremity of this range, that is $\pi/10n$. In particular, one can see from Fig. \ref{fig:error_in_the_inversion} that the errors entailed by the fitting function $g(\theta)$ quickly saturate this worst case scenario, and that their decrease is only due to the size of the inversion ranges considered.

\section{Training a measurement operator using the response function}\label{app:training}
Here we showcase one application of our inference approach for a magnetometry sensing task where
one is given a fixed input probe state and wants to train a measurement operator to obtain the best possible sensitivity. In particular, we will take a variational approach~\cite{cerezo2020variationalreview} to quantum metrology~\cite{meyer2020variational, koczor2020variational,beckey2020variational,kaubruegger2021quantum} where we train a quantum neural network (QNN) to perform the optimal measurement. 

The setting is as follows. First, one prepares a fixed input state. In our case, we prepare a GHZ state on a system of $4$ qubits. Then we proceed by optimizing a quantum convolutional neural network (QCNN)\cite{cong2019quantum} such as to perform the optimal measurement (we chose this QCNN architecture as it is immune to barren plateaus~\cite{pesah2020absence,cerezo2020cost,sharma2020trainability}). In particular, we know that the GHZ state is capable of achieving the Heisenberg limit, meaning that there should exist an optimal measurement (the parity measurement) which achieves such limit.

\begin{figure}[h]
    \centering
    \includegraphics[scale=0.65]{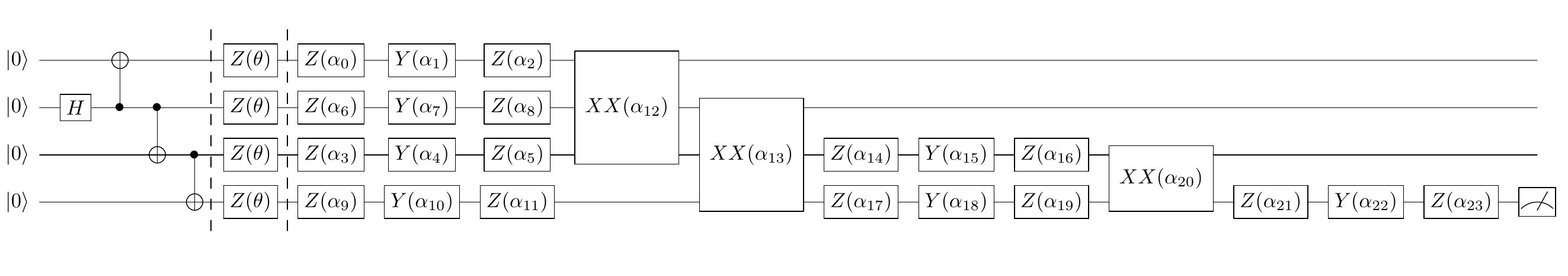}
    \caption{\textbf{Circuit used in the variational quantum metrology numerical experiment}. Four-qubit variational experimental set up where a QCNN is trained to learn the optimal measurement.}
    \label{fig:variational_circuit}
\end{figure}
The circuit diagram for our variational set-up is shown in Fig.~\ref{fig:variational_circuit}. To train the measurement operator we minimize a mean squared error cost function of the form 
\begin{equation}
    \LC_{mse}(\vec{\alpha}) = \frac{n}{2\pi}\int_{-\frac{\pi}{n}}^\frac{\pi}{n}d\theta (\widetilde{R}(\vec{\alpha};\theta) / n -\theta)^2
\end{equation}
where $\widetilde{R}(\vec{\alpha};\theta)$ is the inferred response function which depends on the trainable circuit parameters $\vec{\alpha}$ and $n$ is the number of qubits (here $n=4$). We note that in this cost function we have multiplied the inferred response function by a term of $1/n$ in attempt to enforce the desired Heisenberg scaling. A key feature of this cost is that it can be evaluated analytically given the inferred response function.

In Fig.~\ref{fig:training_qcnn}(a) we show the mean squared error loss function versus the number of epochs, where one can see that the optimizer is indeed able to minimize the loss function. In Fig.~\ref{fig:training_qcnn}(b) we show the sensitivity of the scheme before the QCNN is trained (i.e., for some set of random parameters $\vec{\alpha}$), and we can clearly see that for all parameter values the sensitivity is well above the Heisenberg limit ($1/n^2=1/16\sim0.0625$), and even the Standard limit ($1/4=0.25$). However, after training the QCNN, Fig.~\ref{fig:training_qcnn}(c) shows that the scheme can indeed reach the Heisenberg limit. We note that the simulations where  performed in a noiseless setting where we have not included  the effects of hardware noise or finite sampling. This proof of principle experiment highlights the utility of using the inference based sensing scheme in a variational quantum metrology setting. 

\begin{figure}
    \centering
    \includegraphics{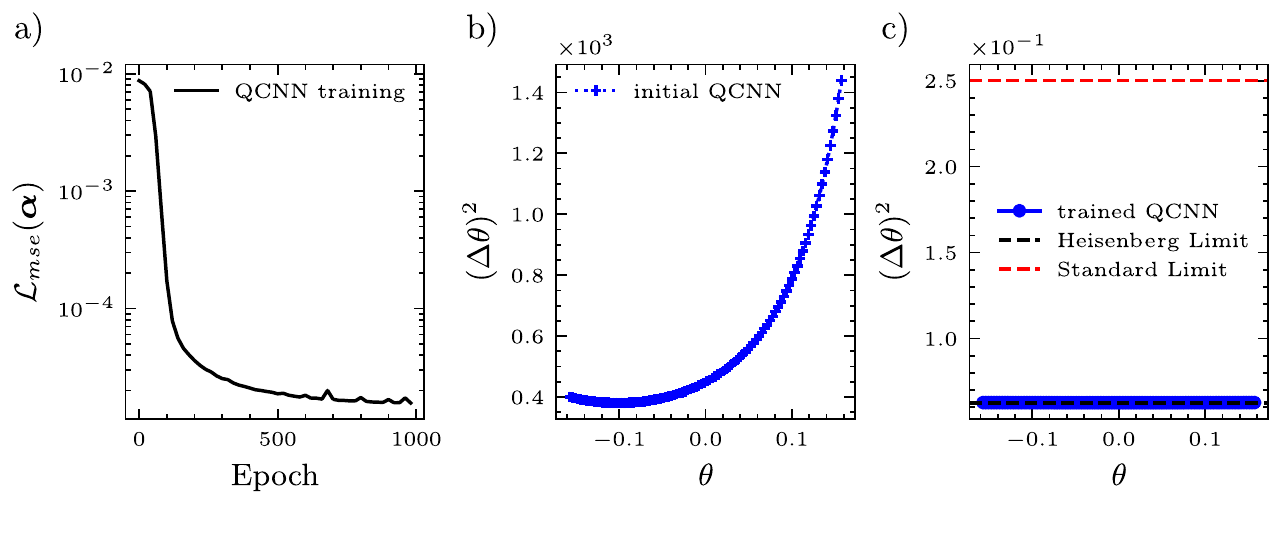}
    \caption{\textbf{Training a QCNN to perform the optimal measurement on a 4 qubit system.} Here we show how one can use a our approach to infer the response function for a QCNN that seeks to prepare the optimal measurement in a magnetometry task. Here the initial state is the GHZ state and we train using the $\LC_{mse}$ which we evaluate analytically. }
    \label{fig:training_qcnn}
\end{figure}

\section{Circuits}\label{app:circuits}
In this section, we display the circuit constructions used in each sensing protocol presented in the experimental and numerical results as described in the main text. For illustrative purposes, we present the circuits for $n=4$ qubits, which readily generalize to larger problem sizes. Vertical lines are used to separate the probe state preparation from the mechanism that encodes the parameter of interest.

In Fig. \ref{fig:circuits}, we present all circuit decompositions used on the manuscript. Figure \ref{fig:circuits}(a)-(b) corresponds to two different circuit decompositions for the GHZ magnetometry problem. It is important to mention that even though they look different, both constructions are equivalent. 
The circuit shown in Fig. \ref{fig:circuits}(a) is commonly used for the ion-trap numerical simulations, while the circuit shown in Fig. \ref{fig:circuits}(b) favors the connectivity and native gates of the \textit{IBM\_Montreal} quantum computer. In the former, we assume an ion-trap quantum computer with full connectivity, enabling a GHZ state preparation in $\mathcal{O}(\log(n))$ depth. Once the GHZ is prepared, a global $Z$ rotation with angle $\theta$, mimicking the effect of a magnetic field, is applied to the probe state followed by a parity measurement. 

For the spin squeezing setup, no step of state preparation is necessary as the coherent spin probe state $\vert 0 \rangle^{\otimes n}$ is the default initialization in current quantum computers. 
In Fig. \ref{fig:circuits}(c), we show the gate decomposition corresponding to the one-axis twisting Hamiltonian mechanism used to characterize squeezing. 
In this particular setting we only need to measure one qubit, which is chosen to be the last one, but measurements on any other qubit would produce the same result. 

Finally, in Fig. \ref{fig:circuits}(d) we show the circuit decomposition for the preparation of a random probe state obtained by application of a random unitary composed of $4$ layers of hardware efficient ansatz. The parameter values of $\{\alpha_{j}, \beta_{i}, \delta_{i}\}$ are randomly drawn. 
After encoding of the parameter $\theta$, each qubit is measured individually. 

\begin{figure}[h]
    \centering
    \includegraphics[scale=1]{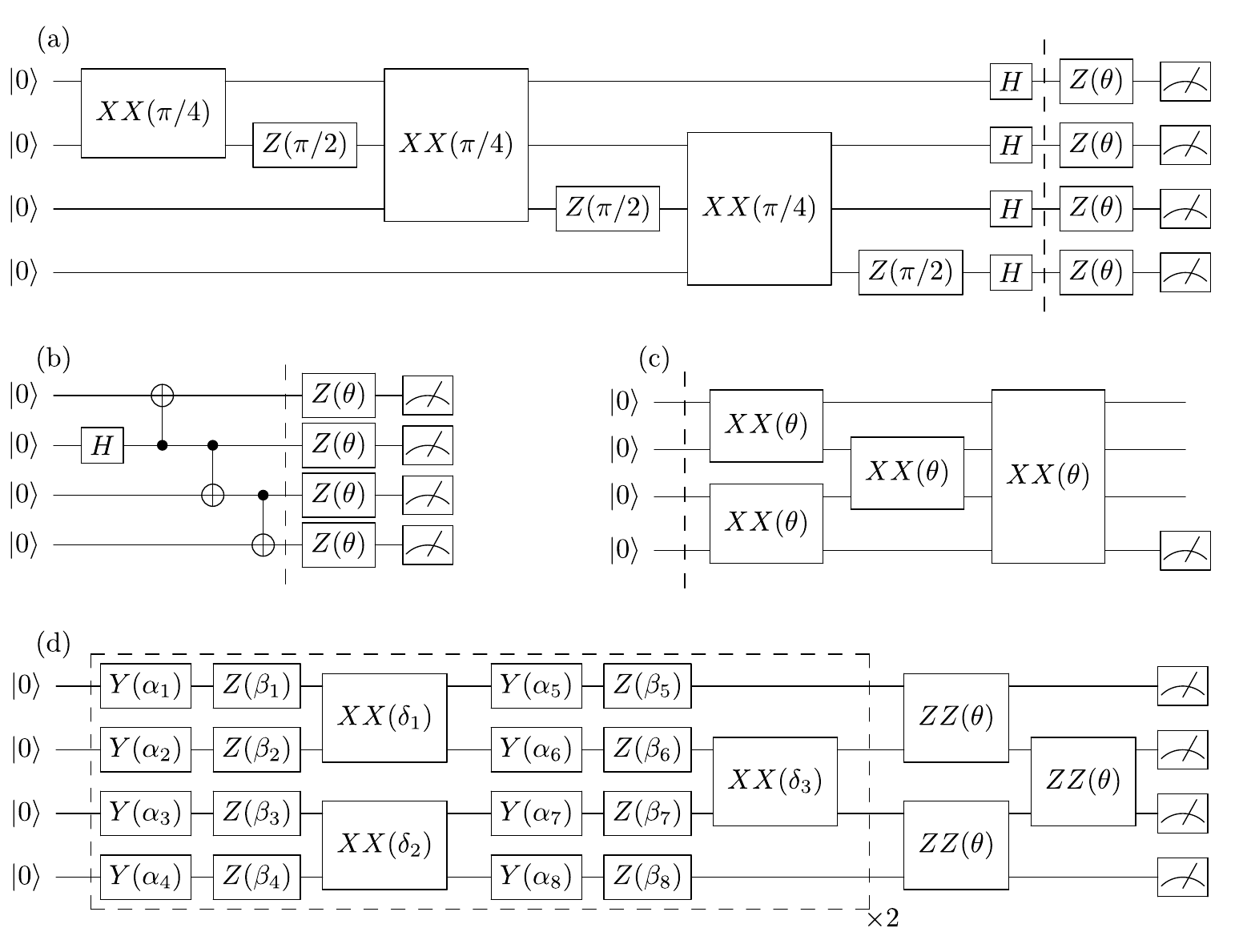}
    \caption{\textbf{Circuit structures used in this work.} Four-qubit circuit decomposition used for: the GHZ magnetometry problem on (a) ion-trap and (b) IBM architectures, (c) spin squeezed setup and (d) the random unitary.}
    \label{fig:circuits}
\end{figure}


\end{document}